\documentclass{article}

\usepackage{latexsym}
\usepackage{amsmath,amssymb,amscd}
\usepackage{graphicx}
\usepackage{xspace}
\usepackage{amsthm,amsmath,amssymb,amscd,amsfonts,latexsym}
\usepackage{bussproofs}
\usepackage{color}
\usepackage{multicol}
\usepackage[english]{babel}
\usepackage{tikz}
\usepackage{authblk}
\usepackage{fancyhdr}
\usepackage{orcidlink}

\usepackage{mathrsfs}
\usepackage{mathtools}
\usepackage{rotating}
\usepackage{natbib}
\usepackage{csquotes}
\usepackage{comment}

\usepackage[inline]{enumitem}
\usepackage{url}
\usepackage{hyperref}

\usepackage{lineno}

\usepackage{newtxmath}

\makeatletter
\def\namedlabel#1#2{\begingroup
    #2%
    \def\@currentlabel{#2}%
    \phantomsection\label{#1}\endgroup
}
\makeatother

\newtheorem{theorem}{Theorem}[section]
\newtheorem{corollary}{Corollary}[theorem]
\newtheorem{definition}{Definition}[section]
\newtheorem{lemma}[theorem]{Lemma}

\usetikzlibrary{matrix}
\usetikzlibrary{positioning}
\usepackage{yfonts}
\usepackage{epsfig}
\usetikzlibrary{arrows,automata}

\newcommand{\Pts}{\sf{PtS}}
\newcommand{\Bes}{\sf{BeS}}
\newcommand{\At}{\sf{At}}

\newcommand{\LJ}{\sf{LJ}}
\newcommand{\LK}{\sf{LK}}

\newcommand{\NEp}{\sf{NEp}}

\newcommand\eg{\hbox{\textit{e.g.}}}

\newcommand\iimp{\to_i}
\newcommand\cimp{\to_c}
\newcommand\ivee{\vee_i}
\newcommand\cvee{\vee_c}

\title{Connecting Base-Extension Semantics, Beth Semantics, and Ecumenical Logic}

\author{Victor Barroso-Nascimento\,\orcidlink{0000-0002-3990-5996}}
\author{Joaquim T. Waddington\,\orcidlink{0009-0009-3362-2467}}
\author{Elaine Pimentel\,\orcidlink{0000-0002-7113-0801}}
\affil{Department of Computer Science, University College London, UK}

\newcommand{\runningtitle}{Ecumenical semantics}

\newcommand{\runningauthors}{Barroso-Nascimento, Waddington, Pimentel}

\date{}

\begin{document}

\maketitle

\begin{abstract}
We investigate the relationship between base-extension semantics, Beth semantics, and Prawitz's ecumenical logic. We first present a base-extension semantics for the propositional fragment of Prawitz's system and prove that it is sound and complete. The semantics interprets intuitionistic and classical operators directly within a common proof-theoretic framework, without reducing the classical operators to their intuitionistic counterparts by means of a translation. We then analyse the effect of incorporating Beth-style clauses into base-extension semantics. When arbitrary bases are admitted and $\bot$ is governed by atomic explosion, these clauses trivialise. By contrast, when $\bot$ is treated as a basic sentence and attention is restricted to consistent bases, the Beth-style clauses for atoms and disjunction characterise, respectively, classical atoms and classical disjunction in Prawitz's system. We also obtain equivalent, path-independent formulations of the classical clauses for atoms, disjunction, and implication. These results reveal that the consistency constraint is not merely a technical restriction: within base-extension semantics, it transforms clauses originally associated with intuitionistic semantics into semantic characterisations of classical operators.
\end{abstract}

\section{Introduction}
The relatively young discipline of proof-theoretic semantics ($\Pts$) provided a fresh breath of air into the semantic study of intuitionistic logic~\cite{sep-proof-theoretic-semantics}. Grounded in inferentialist interpretations of constructivism that regard proofs as semantic ingredients of meaning~\cite{Brandom2000}, $\Pts$ provides new semantic frameworks in which the core semantic values are defined in terms of proof-theoretic structures corresponding to constructive proofs of propositions. Aside from making explicit the relationship between semantics and proofs, $\Pts$ has the additional benefit of allowing use of proof-theoretic instruments for the proof of semantic results.

An important instance of this approach is base-extension semantics ($\Bes$)~\cite{Sandqvist2015-SANBSF}, in which validity is defined relative to systems of rules over basic sentences and their extensions. The resulting semantics exhibits a clear structural affinity with Kripke semantics: bases may be regarded as information states, base extension behaves like an accessibility relation, and the clause for implication quantifies over all extensions of the current base~\cite{Goldfarb2015-GOLODP-2,Piecha2015FailComplete-PIEFOC}. 
This affinity is more than suggestive. Pym, Ritter, and Robinson~\cite{DBLP:journals/sLogica/PymRR25} analyse Sandqvist's
semantics by means of presheaf and sheaf constructions, clarifying its relationship with intuitionistic Kripke models and, in particular, the distinctive behaviour of disjunction in $\Bes$. Correspondences between base-extension and Kripke semantics have also been developed for several classical and intuitionistic modal logics by Eckhardt and Pym~\cite{DBLP:journals/igpl/EckhardtP25} and by Buzoku and Pym~\cite{DBLP:conf/tableaux/BuzokuP25}.

The relationship between $\Bes$ and traditional model-theoretic semantics is nevertheless far from exhausted. In particular, its
connection with Beth semantics~\cite{BethSemantics,BezhanishviliHolliday2019} remains largely unexplored. This is especially intriguing because the two approaches treat disjunction in formally similar ways~\cite{waddington2026inferentialistgamesemanticsextended}. One might
therefore expect Beth's clauses to provide an alternative route to the intuitionistic operators of $\Bes$. However, what happens is almost the opposite! The behaviour of the clauses depends of the treatment of absurdity, but in neither case we have intuitionistic logic.

When Beth-style clauses are transplanted into the usual base-extension setting, in which $\bot$ is understood through atomic explosion, the resulting clauses trivialise: the presence of a greatest base, containing all atomic rules, forces the relevant conditions to hold indiscriminately. The picture changes completely when $\bot$ is itself treated as a basic sentence and the semantics is restricted to consistent bases. Under this constraint, Beth-style clauses cease to be trivial and instead characterise {\em classical} operators. 
Quite surprisingly, this brings Beth semantics into contact with logical ecumenism, which is  an approach to logic that aims at providing frameworks allowing rival logics to peacefully coexist~\cite{Pimentel2025-PIMAEV}. In this setting, Beth's clause for disjunction provides semantics to the classical disjunction of {\em Prawitz's ecumenical system}~\cite{DBLP:journals/Prawitz15}. 

The version of $\Bes$ developed in~\cite{Pimentel2025-PIMAEV} for Barroso-Nascimento's ecumenical system~\cite{Barroso-Nascimento} already points in this direction. There, intuitionistic validity is grounded in the traditional understanding of proofs as constructions, whereas the consistency constraint gives rise to a genuinely classical notion of proof, echoing Hilbert's association between truth and proofs of consistency~\cite{ConsistencyHillbert61554c58-c869-34f0-b322-2cff263d9ae0,Hilbert1979-HILMPL}.

The nature of $\bot$ has long been the subject of discussions in proof-theoretic semantics. In some of his works, Prawitz considers $\bot$ as an object that has no proofs and no introduction rules  \cite{PRAWITZ1971235}. Dummett, on the other hand, defines it in terms of logical explosion \cite{Dummett1991-DUMTLB}. Most versions of base-extension semantics use Dummett's conception, but this paper expands the approach used in \cite{DBLP:journals/synthese/PimentelPP21}, which defines $\bot$ as an unprovable atom. This means that we must understand the meaning of $\bot$ before understanding the logical connectives, but the only content we stipulate for $\bot$ is its unprovability. It is also important to notice that, as shown in this paper and in \cite{DBLP:journals/synthese/PimentelPP21}, both characterisations have different technical properties; for instance, consistency constraints allows entailments to be satisfied vacuously, whereas explosive clauses do not. There also seems to be a trade-off involved in the choice: Dummett's clause seems to be more constructive, as we can actually prove $\bot$ by proving all other propositions, but the consistency constraint makes it so that by proving a negation $\neg A$ we prove that $A$ is never provable, which is arguably our pre-theoretic understanding of negation. In fact, it has been argued that lack of this property might be considered a fault of some standard approaches to proof-theoretic semantics \cite[pg. 324]{Piecha2015FailComplete-PIEFOC}.

The present paper develops this connection for Prawitz's ecumenical logic. Since by using Beth's semantic definitions in this setting we obtain semantic characterisations of Prawitz's classical connectives,  such clauses, originally devised to provide a \textit{genuinely intuitionistic definition of proof}, can also be used to provide a \textit{genuinely classical definition of proof}. Moreover, analogously to what happens in \cite{Pimentel2025-PIMAEV}, the metalanguage of the semantics has to be ecumenical, but classical reasoning principles (which now include set-theoretic as well as logical reasoning) are required \textit{only to establish properties of classical connectives}. Thus, the coexistence between classical and intuitionistic reasoning is reflected not only in the object language, but also in the methods used to study its semantics.

This gives rise to the following motto:

\begin{center}
    {\em Beth-style clauses $+$ consistency constraint $=$ classical operators.}
\end{center}

The contributions of this work are threefold. First, we present a sound and complete base-extension semantics for the propositional fragment of Prawitz's ecumenical system. Its clauses characterise classical support in terms of atomic derivability, semantic consequence and consistent base extension. They validate the familiar negative-translation equivalences and provide the basis for the Beth-style characterisations developed in Section~\ref{sec:Beth}.

Second, we investigate the behaviour of Beth-style clauses over the frames induced by base extension. Over the unrestricted class of bases, with $\bot$ interpreted through atomic explosion, the clauses for atoms and disjunction trivialise. Over the class of consistent bases, however, these same clauses characterise Prawitz's classical atoms and classical disjunction.

Third, we show that these path-based clauses admit equivalent formulations stated entirely in terms of base extension. Besides dispensing with explicit reference to paths, these formulations yield an analogous characterisation of classical implication.

The paper is organised as follows. Section~\ref{sec:overview} recalls the relevant background on logical ecumenism, base-extension semantics, and Beth semantics. Section~\ref{sec:BeS} introduces our base-extension semantics for Prawitz's propositional ecumenical system and establishes soundness and completeness. Section~\ref{sec:Beth} studies the interaction between Beth-style clauses and the frames induced by unrestricted and consistent bases. Finally, Section~\ref{sec:conc} discusses the significance of the results and directions for further work. 

\section{Overview}\label{sec:overview}
Three strands come together in the results of this paper: Prawitz's logical ecumenism, base-extension semantics, and Beth semantics. Each approaches the boundary between classical and intuitionistic reasoning from a different direction: Ecumenism places their operators within a common language; base-extension semantics gives meaning in terms of inferences and their possible extensions; and Beth semantics follows what becomes established along every possible path. In this section, we recall the aspects of these three approaches that will allow them to meet in the technical development that follows.

\subsection{Logical ecumenism}
{\em Logical ecumenism} is a philosophical and proof-theoretic framework introduced by Peter Krauss~\cite{Krauss-P1,Krauss-P2}, Dag Prawitz~\cite{DBLP:journals/Prawitz15} and Gilles Dowek~\cite{DBLP:journals/Dowek16a,DBLP:conf/fscd/BlanquiDGHT21, DBLP:conf/types/Grienenberger22, DBLP:journals/lmcs/BlanquiDGHT23, DBLP:phd/hal/Grienenberger25}, aimed at harmonising opposing logical systems, most notably intuitionistic and classical logic. Rather than viewing constructivism and classicism as mutually exclusive concepts,  ecumenism posits that both can peacefully coexist within a single, unified language. Under this view, logical disagreement is not an intractable dispute over the nature of {\em truth}, but rather a distinction in the logical constants and atomic formulas being expressed. 

At its structural core, an ecumenical system achieves such a harmony by sharing a common set of foundational connectives, while bifurcating others into explicit variants. For example, an intuitionistic implication demands a constructive method to transform proofs, whereas its classical counterpart accommodates non-constructive principles like {\em double negation elimination}. This setup allows classicists and intuitionists to reason in the same workspace while precisely tracking where non-constructive assumptions are deployed.

This is actually the core behind Prawitz's ecumenism~\cite{DBLP:journals/Prawitz15}, where classicists and intuitionists share what are called {\em neutral connectives}: $\wedge,\bot,\forall$ for conjunction, falsity and universal quantification. Thinking of the proof-theoretical nature of these connectives, they all share the same inference rules in. \eg, Gentzen's $\LK$ and $\LJ$ sequent systems~\cite{gentzen1969}. On the other hand, implication, disjunction and the existential quantifier have a quite different behavior in these systems, and this is reflected in Prawitz ecumenism, where there are classical and intuitionistic versions of them. Prawitz's  claim is that they are not, in fact, the same connectives.

Negation deserves a separate discussion. We follow the standard natural-deduction account in which $\neg A$ is introduced by deriving $\bot$ under the assumption $A$. This account admits the usual constructive interpretation: a proof transforms any purported proof of $A$ into a proof of absurdity\footnote{It should be distinguished from Nelson's strong negation, which expresses constructive evidence of falsity and need not coincide with $A\to\bot$~\cite{NelsonConstructiveFalsity,DBLP:conf/mfcs/Barroso-Nascimento26}.}.

As Prawitz puts it:
\begin{quote}
There is no other direct way of getting a ground for asserting $\neg A$ classically, although the classical use of negation differs from the intuitionistic. 
\end{quote}
This naturally raises the question of how the classical use of negation can be captured proof-theoretically. As noted in~\cite{DBLP:journals/Prawitz15}, introducing separate classical and intuitionistic negation constants equipped with their expected inference rules would cause the two connectives to collapse into one, a phenomenon first observed by Popper in~\cite{popper1948deduction2}. While this already provides strong evidence that negation is, in fact, a neutral connective as defended by Prawitz, in this work we present an even stronger proof by showing that classical and intuitionistic negation also collapse at the semantic level (see Section~\ref{sec:BeS}).

Prawitz's formulation of ecumenism is particularly well suited to the present work because it internalises the classical/intuitionistic coexistence in a single deductive system. This makes it a natural testing ground for an inferentialist semantics: If the distinction between its operators is fundamentally a distinction between rules, then their semantic interpretation should reflect those rules directly.

\subsection{Base-extension Semantics}

Whereas model-theoretic semantics~\cite{sep-model-theory} explains
the meaning of logical expressions through their interpretation in
set-theoretic structures, with Tarski's definition of truth as the
paradigmatic example~\cite{Tarski1944}, Proof-theoretic semantics ($\Pts$)~\cite{Prawitz-2006,pts-91,sep-proof-theoretic-semantics}  takes proofs and inferences as
fundamental. Its philosophical basis is {\em inferentialism}~\cite{Brandom2000}: the meaning of a
logical expression is determined by the inferential rules governing it.
Logical constants are therefore characterised not by the conditions under
which they are true in a model, but by the role they play in reasoning.

$\Pts$ has recently attracted considerable attention, see, \eg, \cite{stafford2025logics,DBLP:journals/sLogica/GheorghiuGP26, DBLP:journals/rsl/Ayhan26,Daragona-Prawitz}. By grounding {\em validity} in {\em derivability},  $\Pts$ makes it possible to employ the rich machinery of proof theory in giving meaning to logical operators.

$\Bes$ is a strand of $\Pts$ in which proof-theoretic validity is defined relative to a specified collection of inference rules, called a {\em base}, governing the basic formulas of the language~\cite{Sandqvist2015-SANBSF}.
Such bases provide the underlying proof-theoretic environment from which the semantic behaviour of complex formulas is determined.
These basic formulas are typically atomic propositions, whose meaning is determined by their inferential role within the base.

To illustrate the idea, suppose that one does not know the meaning of the
sentence
\begin{center}
``Giant squids live in the aphotic zone.''
\end{center}
If, however, we are given the inference rules
{\tiny
\begin{prooftree}
 \AxiomC{Giant squids live in the aphotic zone}
 \UnaryInfC{Giant squids live in the ocean}
    \DisplayProof
    \qquad
   \AxiomC{Giant squids live in the aphotic zone}
   \UnaryInfC{Giant squids live in areas with little or no sunlight}
\end{prooftree}
}

\begin{prooftree} \scriptsize
    \AxiomC{Giant squids live in the ocean}
    \AxiomC{Giant squids live in areas with little or no sunlight}
    \BinaryInfC{Giant squids live in the aphotic zone}
\end{prooftree}
then the inferential role of ``aphotic zone'' tells us that, in this setting,
it corresponds to a part of the ocean with little or no sunlight.

Bases do more than specify the inferential roles of basic formulas: they
themselves constitute inference systems in which such formulas can be
derived. For example, let $r,w,u,m,d$ represent, respectively, the sentences
``It's raining'', ``I might get wet'', ``I have an umbrella'', ``I bring the
umbrella with me'', and ``I will stay dry'', and consider a base $S$
containing the following rules
\begin{prooftree}
    \AxiomC{$r$}
    \UnaryInfC{$w$}
    \DisplayProof
    \qquad
    \AxiomC{$w$}
    \AxiomC{$u$}
    \BinaryInfC{$m$}
    \DisplayProof
    \qquad
    \AxiomC{$[r]$}
    \noLine
    \UnaryInfC{$\vdots$}
    \noLine
    \UnaryInfC{$m$}
    \UnaryInfC{$d$}
\end{prooftree}
Then, within $S$, we can derive that ``I will stay dry'' from `` I have an umbrella"
\begin{prooftree}
    \AxiomC{$[r]$}
    \UnaryInfC{$w$}
    \AxiomC{$u$}
    \BinaryInfC{$m$}
    \UnaryInfC{$d$}
\end{prooftree}

The derivation above takes place entirely within the base $S$. A central
feature of $\Bes$, however, is that we do not consider bases in isolation.
A base $S'$ is an {\em extension} of $S$, written $S' \supseteq S$, if it contains all the rules of
$S$, possibly together with additional ones. 
Intuitively, extending a base amounts to enriching our inferential resources:
anything derivable in $S$ remains derivable in $S'$, but the additional
rules may enable new derivations.

For example, consider the extension $S' \supseteq S$  by adding the rule 
\begin{prooftree}
    \AxiomC{$r$}
    \UnaryInfC{$u$}
\end{prooftree}
expressing, for example, that whenever it rains I have an umbrella available.
In the extended base $S'$, from $r$ we can now derive both $w$
and $u$, and hence $m$
\begin{prooftree}
    \AxiomC{$r$}
    \UnaryInfC{$w$}
    \AxiomC{$r$}
    \UnaryInfC{$u$}
    \BinaryInfC{$m$}
\end{prooftree}
The passage from $S$ to $S'$ therefore represents an increase in inferential
information, rather than a change in the meaning of rules already
present in $S$.

This notion of extension plays a fundamental role in $\Bes$. Semantic
support is often required to be stable under extensions of a base: what is
supported given certain inferential resources must remain supported when
those resources are enriched. Moreover, the semantic clauses for some
logical operators explicitly quantify over extensions of the current base.
In this sense, base extensions play a role analogous to accessible future
states in Kripke semantics, although here the underlying structure is
proof-theoretic: what grows is not a state of information specified
model-theoretically, but the collection of available inference rules.

Sandqvist~\cite{Sandqvist2015-SANBSF} gave a sound and complete formulation
of $\Bes$ for intuitionistic propositional logic.
There, the notion of {\em consequence} in $\Bes$ is given by an inductively defined
semantic judgment, called {\em support}, whose base case is determined
by provability in a base. Thus, for any base $S$, the support of an atomic
formula $p$ is defined by
\[
\Vdash_{S} p \quad\mbox{iff}\quad p \mbox{ is provable in } S.
\]
The clauses for complex formulas are then defined inductively. Conjunction,
for example, behaves much as it does in Kripke semantics:
\[
\Vdash_{S} A \land B
\quad\mbox{iff}\quad
\Vdash_{S} A \mbox{ and } \Vdash_{S} B.
\]

Beyond this point, however, the similarity with Kripke semantics begins to
fade. In particular, $\bot$ is often understood as {\em absurdity}, following
its second-order characterisation instantiated over atoms (as in the system
$\mathsf{Fat}$~\cite{DBLP:journals/jsyml/FerreiraF13}):
\[
\Vdash_{S} \bot
\quad\mbox{iff}\quad
\Vdash_{S} p \mbox{ for every atom } p.
\]
Simply requiring $\bot$ to be unsupported, as one might expect by analogy
with Kripke semantics, leads to serious problems. In particular, it allows
one to establish $\Vdash_{S}\neg\neg p$ for every atom
$p$~\cite{Piecha2015FailComplete-PIEFOC}. Perhaps surprisingly, this is not the case if we include $\bot$ in bases but require $\nvdash_{S} \bot$ to hold for all bases. As demonstrated in \cite{Pimentel2025-PIMAEV}, by doing this and using the other semantic clauses of standard $\Bes$ we obtain a semantics for intuitionistic logic.

As we shall show, $\Bes$ provides a common proof-theoretic framework for neutral, intuitionistic and classical operators, with their interaction specified through support clauses over consistent atomic bases and their extensions.

\subsection{$\Bes$ and Beth: Constructivism and Consistency}\label{section: beth}
The impact of definitions of $\bot$ in intuitionistic logic and completeness proofs has a very rich history. Beth's original proof of completeness for his semantics contained an important flaw directly related to the possibility of models being inconsistent \cite[pg. 8]{Troelstra1999-TROTDO-5}. A result due to G{\"o}del and Kreisel shows that this is part of a more general phenomenon, as any completeness proof for semantics having a definition of $\bot$ in terms of unsatisfiability has to validate Markov's principle at the metalinguistic level \cite{GodelKreiselMarkov,Kreiselweak1d8699fd-bfc0-3036-b969-86808e50151e}. On the other hand, Veldman \cite{Weldman3ed2a5f6-4711-3646-ba01-000ad136b830} and de Swart \cite{deSwart5fa3bf73-e0dd-34ae-81ab-659ee7e64c9b} provided intuitionistic completeness proofs for Kripke and Beth first-order intuitionistic semantics by defining $\bot$ in terms of explosion. Sandqvist's  \cite{Sandqvist2015-SANBSF} completeness proofs follow along the same lines.

Beth semantics is defined over ordered frames $F=\langle W,\leq\rangle$, with paths represented by maximal chains.
Beth's approach yields semantics for different logics depending on the class of frames used~\cite{BeerBeth}. Historically, Beth gave particular importance to {\em spreads} and {\em fans}: a spread is a tree in which every node has a successor, and a fan is a finitely branching spread~\cite[p.~679]{TroelstraConstructivismVol2}, \cite{BethSemantics}.
This choice reflects the connection between Beth's frames and Brouwer's free-choice sequences, through which mathematical activity is understood as a process unfolding through successive choices~\cite[p.~21]{BeerBeth}. It also illustrates how both the selection of frames and their intended pre-theoretic interpretation shape the relationship between a semantics and constructive mathematics.

This becomes particularly important for $\Bes$ once we point out that similar readings are applicable to the concepts of \textit{bases} and \textit{extensions}. Given the standard reading of bases as collections of accepted inferences, the extension relation can be visualised as standing for a process of expansion of a knowledge base  \cite{Schroeder-HeisterAtomicTwo2016-SCHASI-16}. If {\em all bases} are allowed in our semantics, we are considering all possible extension process, and if we only allow {\em consistent bases}, we are only considering extensions preserving consistency. 

Despite this conceptual similarity, structural differences makes transplantation of Beth's interpretation to $\Bes$ untenable. This is so since the set of all bases can be characterised as the powerset of the set of all rules \cite{Followingall10.1093/analys/anac100} and, since there are countably infinitely many rules, there are \textit{uncountably many bases}.  This entirely precludes the establishment of any meaningful conceptual relationship between $\Bes$ and choice sequences: while Beth semantics, when defined over fans, is capable of providing infinitary definition with strong, finitary roots, $\Bes$ quantifies over an uncountable set and characterises its extension processes in terms of the \textit{continuum} \footnote{This lends further credence to the argument in \cite[pgs. 127-129]{Nascimentophd} to the effect that, since $\Bes$ forces us to consider all extensions at once in semantic definitions, its treatment of infinity might not be intuitionistically acceptable. Barroso-Nascimento proposes an alternative proof-theoretic treatment of infinity in which the semantics consider only the structure obtained by extending any base finitely many times, then going on to argue that this replaces the notion of completed infinity used in $\Bes$ by one of potential infinity, but this definition does not seem to be compatible with $\Bes$ itself. Although this is of philosophical importance, it does not have any bearing on our technical results, so we do not deal with this issue here. }. This makes it so that the set-theoretic metalinguistic structure employed in $\Bes$ exhibits strong classical features.

Traditional definitions of Beth semantics use a definition of atomic valuation different from that of Kripke semantics \cite{BeerBeth,TroelstraConstructivismVol2}, although it is also possible to define Beth semantics using Kripke valuations \cite{sep-logic-intuitionistic}. The basic difference between both is that Beth valuations allows the value of an atom at a point to affect its value in \textit{antecessors} of that point, whereas Kripke valuations only allows a point to affect its \textit{successors}. The formal definitions are as follows. 

\begin{definition}\label{def:Bethvaluation}
Given any ordered frame $F=\langle W,\leq\rangle$, a function $b : W \mapsto \mathcal{P}(\At)$ is a \textit{Beth valuation} if, for all atoms $p$, if for all maximal chains $\alpha$ such that $w \in \alpha$ there is some $w' \in \alpha$ such that $p \in b(w')$, then $p \in b(w)$.
\end{definition}

\begin{definition}\label{def:Kripkevaluation}
Given any ordered frame $F=\langle W,\leq\rangle$, a function $k : W \mapsto \mathcal{P}(\At)$ is a \textit{Kripke valuation} if $p \in k(w)$ and $w \leq w'$ implies $p \in k(w')$.
\end{definition}

Intuitively, Beth and Kripke valuations establish which atomic facts hold at each point $w$ of the frame. The condition imposed on Beth valuations requires that, whenever an atomic fact is eventually established at some point $w'$ in any maximal chains passing through $w$, it is already be established at $w$. Although the requirement is only stated for atoms, it is possible to prove that the property extends to arbitrary formulas \cite[pg. 25]{BeerBeth}. When combined with fans, Beth valuations can intuitively be read as fixing that everything inevitably knowable in a finite amount of time can already be considered true.  As such, Beth semantics seems to be a particularly good fit for Michael Dummett's constructive concept of truth \cite{DummettTruth}, in which a statement is considered true (although possibly not yet provable or assertible) whenever some procedure leading to a construction of its object is available, even if the construction itself has not been obtained, and regardless of our knowledge of the fact that the procedure will lead to this construction \cite[pg. 31]{Nascimentophd}. It is also worthy of notice that Beth valuations are monotonic, meaning that every Beth valuation is also a Kripke valuation \cite{kripke1965semantical}, but not all Kripke valuations constitute Beth valuations.

Beth's semantic clauses are then defined by using a doubly indexed relation $\vDash^{F}_{w}$ to define validity in $w \in W$ of a frame $F = \langle W, \leq \rangle$, then defining $\vDash^{F}$ in terms of $\vDash^{F}_{w}$ and then $\vDash$ in terms of $\vDash^{F}$. For our purposes it suffices to consider the definitions using $\vDash^{F}_{w}$, as we only compare semantic clause for operators. The clauses for conjunction, implication and entailment correspond to those of Kripke semantics (and $\Bes$), but his clauses for atomic validity and disjunction are as follows:

 \begin{enumerate}
        \item $\vDash^{F}_{w} p$ iff $p \in b(w)$, for atomic $p$;

        \medskip

        \item $\vDash^{F}_{w} A \lor B$ iff for every path $\alpha$ with $w \in \alpha$ there is some $w' \in \alpha$ such that either $\vDash^{F}_{w'} A$ or $\vDash^{F}_{w'} B$;
    \end{enumerate}

Although the definition of atomic support in $\Bes$ seems to be closer to Kripke's definition of atomic validity, Beth's definition of disjunction is remarkably similar to that of $\Bes$. In particular, both allow a disjunction $A \lor B$ to be semantically valid in a point $w$ (or base $S$) even when neither $A$ nor $B$ are valid in that particular point. If we change Beth's clause to one in which $A \lor B$ is valid in $w$ iff either $A$ or $B$ is valid in $w$, we get the intuitionistic Kripke semantic, but if the same change is done in $\Bes$ this leads to incompleteness w.r.t. intuitionistic logic \cite{Piecha2015FailComplete-PIEFOC}. Since the semantic clauses of $\Bes$ can be interpreted as definitions over the frame containing all bases, ordered by set-theoretic inclusion \cite{Goldfarb2015-GOLODP-2,Piecha2015FailComplete-PIEFOC}, this result can be interpreted as a consequence of the fact that the frames use in $\Bes$ are much more restrictive than those used in Beth and Kripke semantics.

 It is important to note that, unlike Kripke semantics, Beth semantics is highly frame-sensitive \cite{BeerBeth}. To provide an example, while the set of finitary Kripke models is sound and complete w.r.t. intuitionistic logic, every finite Beth model is classical \cite[pg. 42]{BeerBeth}. The results we will prove are intimately connected to the frame structures implicitly defined by $\Bes$ and the behaviour they induce in Beth clauses. Since standard $\Bes$ with $\bot$ defined explosively and $\Bes$ with the consistency constraint implicitly define different frames, pairing Beth clauses with them has different effects. In short, the connections we will show between Beth semantics and standard $\Bes$ hold since the set of all bases, when ordered by the set-inclusion relation, has a greatest element (the set of all rules), while the set of all \textit{consistent} bases, despite not having a greatest element, contains maximal consistent bases as maximal elements in all its maximal chains.
 
 \section{$\Bes$ for Prawitz's ecumenical system}\label{sec:BeS}
In this work, we adopt
Sandqvist's terminology and adapt the framework to the ecumenical setting,
following the developments in
\cite{Pimentel2025-PIMAEV,DBLP:conf/tableaux/BarrosoNascimentoPP25}. In particular, we treat $\bot$
itself as a basic formula while restricting attention to consistent bases.

The  propositional {\em base language} is assumed to have the set $\At$ consisting of a countable infinite set of atomic propositions $\{p_1, p_2,\ldots\}$, together with the unit $\bot$.
The elements of $\At$ will be called {\em basic sentences}. We use, as does Sandqvist, systems containing natural deduction rules over basic sentences for the semantical analysis, and we allow inference rules to discharge sets of basic hypotheses. Sets used in the definition of the derivability relation and semantic consequence are always assumed to be finite. Unlike Sandqvist, however, we allow the logical constant $\bot$ to be manipulated by the rules. 

In what follows, we will adopt the following notation: capital Greek letters with the subscript `$\At$' ($\Gamma_{\At}, \Delta_{\At}$) denote finite sets of basic sentences; capital letters ($A, B$) denote formulas; capital Greek letters without the subscript `$\At$' ($\Gamma, \Delta$) denote sets of formulas; commas between sets denote set union.

\begin{definition}[Base] \label{def:atomic_system}
An {\em atomic system} (a.k.a. {\em base}) $S$ is a (possibly empty) set of atomic rules of the form
\begin{prooftree}
   \AxiomC{ $[\Delta^1_{\At}]$}
    \noLine
    \UnaryInfC{$\vdots$}
    \noLine
    \UnaryInfC{$p_1$}    \AxiomC{\ldots}\AxiomC{$[\Delta^{n}_{\At}]$}
    \noLine
    \UnaryInfC{$\vdots$}
    \noLine
    \UnaryInfC{$p_n$}
    \TrinaryInfC{$p$}
\end{prooftree}
where $p_j, p \in \At$ and $\Delta^j_{\At}$ are (possibly empty) finite sets of basic sentences.The sequence $\langle p_1, \ldots, p_n\rangle$ of premises of the rule can be empty -- in this case the rule is called an {\em atomic axiom}. Labels will sometimes be written as the superscript of $\Delta^j_{\At}$ and to the right of a rule to denote that $\Delta^j_{\At}$ was discharged at that rule application.
\end{definition}

\begin{definition}[Extension]
\label{def:extension}
An atomic system $S'$ is an {\em extension} of an atomic system $S$ (written $S' \supseteq S$), if $S'$ results from adding a (possibly empty) set of atomic rules to $S$.
\end{definition}

\begin{definition}[Deducibility]
The {\em deducibility relation} $\vdash_S$ coincides with the usual notion in the system of natural deduction consisting of just the rules in $S$, that is, $p_1, \ldots, p_n \vdash_S p$ iff there exists a deduction with the rules of $S$ whose conclusion is $p$ and whose set of undischarged premises is a subset of $\{p_1, \ldots, p_n\}$.
\end{definition}
Note that, following this definition, if $p \in \{p_1, \ldots, p_n\}$ then $p_1, \ldots, p_n\vdash_S p$ holds regardless of  $S$, as a single occurrence of the assumption $p$ already counts as a ruleless deduction of $p$ from $p$.

\begin{definition}[Consistency]\label{def:deduc}
An atomic system $S$ is called {\em consistent} if $\nvdash_S \bot$.
\end{definition}
All the bases considered in this section are consistent. In Section~\ref{sec:Beth} we will discuss an alternative semantics where this restriction can be dropped.

Prawitz's ecumenical propositional base language $\mathcal{L}$~\cite{DBLP:journals/Prawitz15} is described as follows. We will use a subscript $c$ for the classical meaning and $i$ for the intuitionistic one, dropping such subscripts when formulas/connectives can have either meaning. 
$\bot\in\mathcal{L}$ and, if $p\in\At-\{\bot\}$, then $p_{c}, p_{i}\in\mathcal{L}$\footnote{Although one could initially distinguish between classical and intuitionistic notions of $\bot$, it is straightforward to show that they collapse, as observed in~\cite{DBLP:journals/Prawitz15,DBLP:journals/Dowek16a}. For clarity, we adopt the intuitionistic interpretation throughout and write $\bot$ for $\bot^i$.}.
The neutral logical connectives $\{\neg,\wedge\}$ are shared between the classical and intuitionistic fragments, while $\{\cimp,\cvee\}$ and $\{\iimp,\ivee\}$ are restricted to classical  and intuitionistic interpretations, respectively.

Prawitz's natural deduction system for the propositional fragment of ecumenical logic, denoted by $\NEp$, is presented in Figure~\ref{Fig:NEp}. We write $\Gamma\vdash_{\NEp}A$ when there is a $\NEp$-deduction with conclusion $A$ whose undischarged assumptions form a subset of $\Gamma$.

\begin{figure}[tp] 
{\sc Intuitionistic and neutral rules}
\begin{prooftree}
\AxiomC{$[A]$}
\noLine
\UnaryInfC{$\Pi$}
\noLine
\UnaryInfC{$B$}
\RightLabel{\scriptsize{$\to_i$-int}}
\UnaryInfC{$A \to_i B$}
\DisplayProof
\qquad
\AxiomC{$A \to_i B$}
\AxiomC{$A$}
\RightLabel{\scriptsize{$\to_i$-elim}}
\BinaryInfC{$B$}
\end{prooftree}

\begin{prooftree}
\AxiomC{$A_j$}
\RightLabel{\scriptsize{$\lor_i$-int}}
\UnaryInfC{$A_1\lor_i A_2$}
\DisplayProof
\qquad
\AxiomC{$A \lor_i B$}
\AxiomC{$[A]$}
\noLine
\UnaryInfC{$\Pi_1$}
\noLine
\UnaryInfC{$C$}
\AxiomC{$[B]$}
\noLine
\UnaryInfC{$\Pi_2$}
\noLine
\UnaryInfC{$C$}
\RightLabel{\scriptsize{$\lor_i$-elim}}
\TrinaryInfC{$C$}
\end{prooftree}

\begin{prooftree}
\AxiomC{$A$}
\AxiomC{$B$}
\RightLabel{\scriptsize{$\wedge$-int}}
\BinaryInfC{$A \wedge B$}
\DisplayProof
\qquad
\AxiomC{$A_1\wedge A_2$}
\RightLabel{\scriptsize{$\wedge$$-elim_j$}}
\UnaryInfC{$A_j$}
\end{prooftree}
\begin{prooftree}
\AxiomC{$\bot$}
\RightLabel{\scriptsize{$\bot$-elim}}
\UnaryInfC{$A$}
\DisplayProof
\qquad
            \AxiomC{$[A]$}
            \noLine
            \UnaryInfC{$\Pi$}
            \noLine
            \UnaryInfC{$\bot$}
            \RightLabel{\scriptsize{$\neg$-int}}
            \UnaryInfC{$\neg A$}
\DisplayProof
\qquad            
\AxiomC{$A$}
          \AxiomC{$\neg A$}
             \RightLabel{\scriptsize{$\neg$-elim}}
            \BinaryInfC{$\bot$}
        \end{prooftree}

{\sc Classical rules}
      \begin{prooftree}
            \AxiomC{$[\neg p_ i]$}
            \noLine
            \UnaryInfC{$\Pi$}
            \noLine
            \UnaryInfC{$\bot$} 
            \RightLabel{\scriptsize{$p_c$-int}} 
            \UnaryInfC{$p_c$}
\DisplayProof
\qquad            
            \AxiomC{$p_ c$}
            \AxiomC{$\neg p_i$}
            \RightLabel{\scriptsize{$p_c$-elim}}
            \BinaryInfC{$\bot$}         
        \end{prooftree}
        \begin{prooftree}
            \AxiomC{$[\neg A]$}
            \AxiomC{$[\neg B]$}
            \noLine
            \BinaryInfC{$\Pi$}
            \noLine
            \UnaryInfC{$\bot$} 
            \RightLabel{\scriptsize{$\lor_c$-int}} 
            \UnaryInfC{$A \lor_c B$}
\DisplayProof
\qquad            
\AxiomC{$A \lor_cB$}
            \AxiomC{$\neg A$}
            \AxiomC{$\neg B$}
            \RightLabel{\scriptsize{$\lor_c$-elim}}
            \TrinaryInfC{$\bot$}         
        \end{prooftree}
        \begin{prooftree}
            \AxiomC{$[A]$}
            \AxiomC{$[\neg B]$}
            \noLine
            \BinaryInfC{$\Pi$}
            \noLine
            \UnaryInfC{$\bot$} 
           \RightLabel{\scriptsize{$\to_c$-int}} 
            \UnaryInfC{$A \to_c B$}
\DisplayProof
\qquad            
            \AxiomC{$A \to_cB$}
            \AxiomC{$A$}
            \AxiomC{$\neg B$}
            \RightLabel{\scriptsize{$\to_c$-elim}}
            \TrinaryInfC{$\bot$}
\end{prooftree}
\caption{Ecumenical natural deduction system $\NEp$.}\label{Fig:NEp}
\end{figure}

An ecumenical version of $\Bes$ is given by the definition below. As usual, we start by giving semantic conditions for basic sentences in atomic systems,  expanding them through semantic clauses, having as a starting point Sandqvist's semantics for intuitionistic logic~\cite{Sandqvist2015-SANBSF}.

\begin{definition}[Validity]\label{def:validity} 
{\em Ecumenical $S$-validity ($\Vdash_S$)} and {\em validity ($\Vdash$)} are defined by the following clauses.
\begin{enumerate}\label{def:support}
    \item $\Vdash_S p_i$ iff $\vdash_S p$, for $p \in \At$;\label{Atm}

    \item  $\Vdash_S p_c$ iff $\forall S' (S' \supseteq S) : p \nvdash_{S'} \bot$, for $p \in \At$;\label{Atmc}

    \item  $\Vdash_S A \land B$ iff $\Vdash_S A$ and $\Vdash_S B$;\label{and}

    \item  $\Vdash_{S} A \lor_i B$ iff $\forall S' (S' \supseteq S)$ and all $p \in \At$, $A \Vdash_{S'}p_i$ and $B \Vdash_{S'} p_i$ implies $\Vdash_{S'} p_i$;\label{ori}

    \item  $\Vdash_{S} A \to_i B$ iff $A \Vdash_S B$;\label{toi}
    
    \item $\Vdash_{S} A \lor_cB$ iff $\forall S' (S' \supseteq S)$, $A \Vdash_{S'} \bot$ and $B\Vdash_{S'} \bot$ implies $\Vdash_{S'} \bot$;\label{orc}

    \item$\Vdash_{S} A \to_c B$ iff $\forall S' (S' \supseteq S)$, $\Vdash_{S'} A$ and $B\Vdash_{S'} \bot$ implies $\Vdash_{S'} \bot$;\label{toc}

    \item For non-empty $\Gamma$, $\Gamma \Vdash_S A$ iff for all $S'$ such that $S' \supseteq S$ it holds that, if $\Vdash_{S'}B$ for all $B \in \Gamma$, then $\Vdash_{S'} A$;\label{Inf}

    \item $\Gamma \Vdash A$ iff $\Gamma \Vdash_S A$ for all $S$.\label{Val}
\end{enumerate}
\end{definition}
A distinctive feature of Prawitz's system is that negation is shared by the classical and intuitionistic fragments. This may initially seem surprising, since negation is often regarded, together with implication, as one of the principal points of disagreement between classical and intuitionistic logic. Indeed, if negation is understood as \emph{implication into falsity}, then the presence of two implications might appear to call for two corresponding notions of negation.

Prawitz~\cite{DBLP:journals/Prawitz15} argues, however, that classical and intuitionistic negation do not merely happen to be interderivable~\cite{DBLP:journals/synthese/PimentelPP21}. They have the same grounds for assertion: In either setting, $\neg A$ is established by deriving $\bot$ from the assumption $A$. Since the two negations would therefore have the same introduction rule, they should be understood as one and the same connective. What differs is not the meaning of negation, but the circumstances in which reductio ad absurdum may be applied. In the ecumenical system, classical reductio is justified when its conclusion is classical, but not, in general, when the conclusion contains intuitionistic operators. The distinction between classical and intuitionistic reasoning is therefore carried by the formulas to which negation is applied and by the permitted uses of reductio, rather than by two different negation constants.

Here we provide independent semantic evidence that the two negations coincide.
Indeed, by Clause~\ref{toc}, $\Vdash_S A\to_c\bot$  holds iff, for every $S'\supseteq S$, if $\Vdash_{S'}A$ and $\bot\Vdash_{S'}\bot$ then $\Vdash_{S'}\bot.$
Since $\bot\Vdash_{S'}\bot$ always holds, this is equivalent to requiring that, for every $S'\supseteq S$, if
$\Vdash_{S'}A$ then $\Vdash_{S'}\bot.$
By Clause~\ref{Inf}, this condition is precisely $A\Vdash_S\bot$. But this is the exact definition of $\Vdash_S A\to_i\bot$ given by Clause~\ref{toi}. 
Thus, classical and intuitionistic implication yield exactly the same semantic notion of negation, that is
\begin{enumerate}
\item[{\em 5'.}] $\Vdash_{S} \neg A $ iff $A \Vdash_S \bot$;\label{neg}
\end{enumerate}

We end this section by highlighting the role played by $\bot$ in our semantics. As in~\cite{Pimentel2025-PIMAEV,DBLP:conf/tableaux/BarrosoNascimentoPP25}, $\bot$ belongs to the base language and may therefore occur in atomic rules. At the same time, it plays a central role in the semantic characterisation of the classical fragment. While neutral and intuitionistic formulas are supported directly by their derivability in a base, the support of a classical formula is characterised in terms of consequence to $\bot$.
However, the semantic clauses themselves remain proof-theoretic: support is defined entirely in terms of derivability in atomic bases, their extensions, and the consequence relation induced by them. In particular, no independently given notion of truth or model-theoretic interpretation is invoked in order to distinguish the classical and intuitionistic fragments.

In Section~\ref{sec:Beth}, we explore an alternative formulation in which the meaning of the classical connectives is given without appealing to $\bot$. This alternative essentially depends on use of the consistency constraint, as when combined with the traditional definition of $\Bes$ they yield only trivial connectives.

\subsection{Basic Lemmata}
In this section we prove some properties of our semantics. We omit some proofs for the intuitionistic and neutral connectives when they already appear in \cite{Sandqvist2015-SANBSF}.

We begin by showing that consistency is preserved when adding non-contradictory atomic axioms.
\begin{lemma} \label{pc lemma}
Let $S$ be a base, $p\in\At$ and $S'\supseteq S$ obtained by adding an atomic axiom with conclusion $p$.
Then $p \nvdash_S \bot$ iff $S'$  is consistent.
\end{lemma}
\begin{proof} ($\Rightarrow$) 
Assume $p \nvdash_S \bot$ and suppose, by contradiction, that $\vdash_{S'} \bot$. Replacing every occurrence of a newly added
axiom with conclusion $p$ by an open assumption $p$ yields a
deduction of $p\vdash_S\bot$ using only rules from $S$, which is a contradiction. Thus $\nvdash_{S'} \bot$ and hence $S'$ is consistent. 

\noindent
($\Leftarrow$) Assume $S'$ consistent and $p \vdash_{S} \bot$. Since $S' \supseteq S$ we have $p\vdash_{S'} \bot$ and, together with $\vdash_{S'}p$, this implies that $\vdash_{S'} \bot$, thus violating the consistency requirement. Hence, $p \nvdash_S \bot$.
\end{proof}
We next show that the perfect match between derivability and semantics established by Clause~\ref{Atm} of Definition~\ref{def:validity} extends to contexts as well.
\begin{lemma} \label{Lemma 1}
    Let $S$ be any atomic system. Then $\Gamma_{\At} \Vdash_S p_i$ iff $\Gamma_{\At} \vdash_S p$ or $\Gamma_{\At} \vdash_S \bot$.  
\end{lemma}
\begin{proof} ($\Rightarrow$) Assume that $\Gamma_{\At} \Vdash_S p_i$, and let $S'$ be obtained from $S$ by adding, for each $q\in\Gamma_{\At}$, an atomic axiom
with conclusion $q$. 

Suppose first that $S'$ is consistent. By construction,
$\vdash_{S'}q$, and hence
$\Vdash_{S'}q^i$ for every $q\in\Gamma_{\At}$ by Clause~\ref{Atm}. Since
$\Gamma_{\At}\Vdash_S p_i$, Clause~\ref{Inf} gives
$\Vdash_{S'}p_i$. By Clause~\ref{Atm}, it follows that
$\vdash_{S'}p$. 

Now replace every occurrence in this deduction of a newly added
axiom with conclusion $q$ by an open assumption $q$. The resulting
deduction uses only rules from $S$ and derives $p$ from assumptions
belonging to $\Gamma_{\At}$. Therefore, $\Gamma_{\At}\vdash_S p$.

Suppose instead that $S'$ is inconsistent. Then
$\vdash_{S'}\bot$. Replacing every occurrence of a newly added
axiom with conclusion $q$ by an open assumption $q$ yields a
deduction using only rules from $S$, deriving $\bot$ from assumptions
belonging to $\Gamma_{\At}$. Hence $\Gamma_{\At}\vdash_S\bot$.

\noindent ($\Leftarrow$) 
Suppose first that $\Gamma_{\At}\vdash_S p$. Let $S'\supseteq S$
be any atomic system such that $\Vdash_{S'}q_i$ for every
$q\in\Gamma_{\At}$. By Clause~\ref{Atm},
$\vdash_{S'}q$ for every $q\in\Gamma_{\At}$. Moreover, since
derivability is preserved under extensions of the atomic system, $\Gamma_{\At}\vdash_{S'}p$. \\
Substituting the closed deductions of the members of
$\Gamma_{\At}$ for the corresponding open assumptions yields
$\vdash_{S'}p$. Thus, by Clause~\ref{Atm}, $\Vdash_{S'}p_i$.
Since $S'$ was arbitrary, Clause~\ref{Inf} gives
$\Gamma_{\At}\Vdash_S p_i$. \\
Finally, suppose that $\Gamma_{\At}\vdash_S\bot$. If there were an
atomic system $S'\supseteq S$ such that $\Vdash_{S'}q_i$ for every
$q\in\Gamma_{\At}$, then Clause~\ref{Atm} would give
$\vdash_{S'}q$ for every $q\in\Gamma_{\At}$. Since
$\Gamma_{\At}\vdash_S\bot$ and $S'\supseteq S$, substituting these
closed deductions for the open assumptions would yield
$\vdash_{S'}\bot$, contradicting the consistency of $S'$.
Consequently, there is no such extension $S'$, and the condition in
Clause~\ref{Inf} holds vacuously. Therefore,
$\Gamma_{\At}\Vdash_S p_i$.     
 \end{proof}

We now prove the standard result of monotonicity over bases.
\begin{lemma}[Monotonicity] \label{monotonicity}
If $\Gamma \Vdash_S D$ and $S' \supseteq S$ then $\Gamma \Vdash_{S'} D$.
\end{lemma} 

\begin{proof}
Suppose first that $\Gamma = \varnothing$. The proof is by case analysis on the structure of $D$.
\begin{enumerate}
    \item
     Assume that $\Vdash_S p_c$ and take any $S'' \supseteq S' \supseteq S$. Since $S''$ is, itself, an extension of $S$, by Clause~\ref{Atmc}, $p\nvdash_{S''} \bot$. Since $S''$ is an arbitrary extension of $S'$, Clause~\ref{Atmc} implies $\Vdash_{S'} p_c$.

    \item Assume that $\Vdash_S A \lor_c B$, and let $S'' \supseteq S' \supseteq S$ be such that $A \Vdash_{S''} \bot$ and $B \Vdash_{S''} \bot$. Since $S''$ is an extension of $S$, Clause~\ref{orc} yields $\Vdash_{S''}\bot$. As $S''$ was an arbitrary extension of $S'$ satisfying these conditions, another application of Clause~\ref{orc} yields $\Vdash_{S'} A \lor_c B$.

    \item Assume that $\Vdash_SA \to_c B$, and let $S'' \supseteq S' \supseteq S$ be such that $\Vdash_{S''} A$ and $B\Vdash_{S''} \bot$. Since $S''$ is an extension of $S$, Clause~\ref{toc} yields $\Vdash_{S''}\bot$. As $S''$ was an arbitrary extension of $S'$ satisfying these conditions, another application of Clause~\ref{toc} yields $\Vdash_{S'} A \to_c B$. 
\end{enumerate}
Now suppose that $\Gamma\neq\varnothing$, and let $S''\supseteq S'\supseteq S$ be any extension such that $\Vdash_{S''}C$ for every $C\in\Gamma$. Since $\Gamma\Vdash_S D$, Clause~\ref{Inf} yields $\Vdash_{S''} D$. As $S''$ was an arbitrary extension of $S'$ satisfying $\Vdash_{S''}C$ for every $C\in\Gamma$, another application of Clause~\ref{Inf} gives $\Gamma\Vdash_{S'} D$.

\end{proof}

The following result generalises the disjunction clause to arbitrary formulas.
\begin{lemma}[Disjunction Lemma] \label{disjunction lemma}
If $\Vdash_S A \lor_i B$, $A \Vdash_S E$ and $B \Vdash_S E$, then $\Vdash_S E$
\end{lemma}
\begin{proof} The proof is by case analysis on the structure of $E$. 

\begin{enumerate}
    \item If $E = p_i$ for some $p \in \At$, then the result follows immediately from Clause~\ref{ori}.
    
    \item  Assume that $\Vdash_S A\lor_i B$, $A\Vdash_S p_c$, and $B\Vdash_S p_c$. Suppose, by contradiction, that there exists  $S'\supseteq S$ such that $p\vdash_{S'}\bot$.

We first show that $\nVdash_{S''}A$ for every $S''\supseteq S'$. Indeed, if $\Vdash_{S''}A$ for some such $S''$, then $A\Vdash_S p_c$ yields $\Vdash_{S''}p_c$, and hence $p\nvdash_{S''}\bot$. On the other hand, since $p\vdash_{S'}\bot$ and $S''\supseteq S'$, monotonicity of derivability gives $p\vdash_{S''}\bot$, a contradiction. An analogous argument shows that $\nVdash_{S''}B$ for every $S''\supseteq S'$.

It follows vacuously that $A\Vdash_{S'}\bot$ and $B\Vdash_{S'}\bot$. Since $\Vdash_S A\lor_i B$, the clause for intuitionistic disjunction yields $\Vdash_{S'}\bot$,  contradicting the consistency of $S'$.

Thus, $p\nvdash_{S'}\bot$ for every extension $S'\supseteq S$, and hence $\Vdash_S p_c$.

    \item Assume first that $A \Vdash_S C \lor_c D$. Take $S' \supseteq S$ with $C \Vdash_{S'} \bot$ and $D \Vdash_{S'} \bot$, and take $S'' \supseteq S'$ with $\Vdash_{S''} A$. Thus, by monotonicity, $C \Vdash_{S''} \bot$ and $D \Vdash_{S''} \bot$ and hence Clause~\ref{Inf} yields $\Vdash_{S''} C \lor_c D$. Now, using Clause~\ref{orc}, we conclude that $\Vdash_{S''} \bot$. Since $S''$ is an arbitrary extension of $S'$ with $\Vdash_{S''} A$,  Clause~\ref{Inf} implies that $A \Vdash_{S'} \bot$. A similar argument shows that $B \Vdash_{S'} \bot$. 
    
Since $\Vdash_{S'}A\lor_i B$, $A\Vdash_{S'}\bot$, and $B\Vdash_{S'}\bot$, 
the clause for intuitionistic disjunction yields $\Vdash_{S'}\bot$. As $S'$ is an arbitrary extension of $S$, it follows that, for every $S'\supseteq S$, if $C\Vdash_{S'}\bot$ and $D\Vdash_{S'}\bot$, then $\Vdash_{S'}\bot$. Therefore, by Clause~\ref{orc}, $\Vdash_S C\lor_c D$.

    \item Assume that $A \Vdash_S C \to_c D$ and take any $S' \supseteq S$ with $\Vdash_{S'} C$ and $D \Vdash_{S'} \bot$. Now take any $S'' \supseteq S'$ with $\Vdash_{S''} A$. Thus, by monotonicity, $\Vdash_{S''} C$ and $D \Vdash_{S''} \bot$, and by Clause~\ref{Inf}, $\Vdash_{S''} C \to_c D$. Now we conclude that $\Vdash_{S''} \bot$ by Clause~\ref{toc}. Since $S''$ is an arbitrary extension of $S'$ with $\Vdash_{S'} A$, Clause~\ref{Inf} yields $A \Vdash_{S'} \bot$. A similar argument shows that $B \Vdash_{S'} \bot$. 
    
    Since $\Vdash_{S'} A \lor_i B$, $A \Vdash_{S'} \bot$ and $B \Vdash_{S'} \bot$, $\Vdash_{S'} \bot$ by definition. As $S'$ is an arbitrary extension of $S$, it follows that, for every $S'\supseteq S$, $\Vdash_{S'} C$ and $D \Vdash_{S'} \bot$ implies $\Vdash_{S'} \bot$. Hence, by Clause~\ref{toc},  $\Vdash_S C \to_c D$.    
\end{enumerate}
\end{proof}

The rest of this section is devoted to proving that $\NEp$ is sound and complete w.r.t. the proposed $\Bes$. We start with the soundness result.
\subsection {Soundness}
\begin{theorem}[Soundness]
    If $\Gamma \vdash_{\NEp} A$ then $\Gamma \Vdash A$.
\end{theorem}
\begin{proof}
The proof is by induction on the length $m$ of the derivation $\Pi$ of $\Gamma \vdash_{\NEp} A$.

\textbf{Base case ($m=0$).} $\Pi$ consists in a single occurrence of a proof assumption $A$, hence being a deduction showing $\Gamma \vdash_{\NEp} A$ where $A\in\Gamma$. Since $A \Vdash_S A$ holds for every $S$, we  conclude that $\Gamma \Vdash A$ by Clause~\ref{Inf}.


\textbf{Inductive case ($m>0$)}. We will omit the proof of the intuitionistic and neutral connectives as they are exactly the same as in \cite{Sandqvist2015-SANBSF}. The result is proved by considering the last rule applied in $\Pi$, showing that each rule is semantically valid.
   
\begin{enumerate}
    \item[$p_c$-int]
        Assume, by inductive hypothesis, that $\Gamma, \neg p_i \Vdash \bot$. Let $S$ be an arbitrary atomic system. By Clause~\ref{Inf}, we have $\Gamma, \neg p_i\Vdash_S \bot$. Now let $S'\supseteq S$ be such that $\Vdash_{S'}C$ for every $C\in\Gamma$. Another application of Clause~\ref{Inf} yields $\neg p_i\Vdash_{S'} \bot$. We claim that $p\nvdash_{S''}\bot$ for every $S''\supseteq S'$. Suppose, by contradiction, that $p\vdash_{S''}\bot$ for some $S''\supseteq S'$. By Lemma~\ref{Lemma 1}, it follows that
        $p_i \Vdash_{S''} \bot$. Hence, by Clause~\ref{toi},  $\Vdash_{S''} \neg p_i$. Moreover, since $\neg p_i\Vdash_{S'}\bot$ and $S''\supseteq S'$, monotonicity gives $\neg p_i\Vdash_{S''} \bot$. Therefore, $\Vdash_{S''}\bot$, contradicting the consistency of $S''$. Thus, $p_i\nvdash_{S''}\bot$ for every $S''\supseteq S'$. By Clause~\ref{Atmc}, we conclude that $\Vdash_{S'}p_c$.
        Since $S'$ is an arbitrary extension of $S$ validating every formula in $\Gamma$, Clause~\ref{Inf} yields $\Gamma\Vdash_Sp_c$. Finally, since $S$ is arbitrary, we conclude that $\Gamma \Vdash p_c$ by Clause~\ref{Val}.
    \item[$p_c$-elim]
        Assume, by inductive hypothesis, that $\Gamma_1 \Vdash \neg p_i$ and  $\Gamma_2 \Vdash p_c$. Then, by Clause~\ref{Inf}, $\Gamma_1 \Vdash_S \neg p_i$ and $\Gamma_2 \Vdash_S p_c$ hold for any arbitrary base $S$. Now, take a base $S' \supseteq S$ such that $\Vdash_{S'} C$ for every $C \in \Gamma_1 \cup \Gamma_2$. Then, by the fact that $\Gamma_1 \Vdash_S \neg p_i$ and Clause~\ref{Inf} we conclude that $\Vdash_{S'} \neg p_i$. Similarly, by the the fact that $\Gamma_2 \Vdash_S p_c$ and Clause~\ref{Inf}, we conclude that $\Vdash_{S'} p_c$. By Clause~\ref{toi}, we have that $p_i \Vdash_{S'} \bot$ and by Lemma~\ref{Lemma 1} we conclude that $p \vdash_{S'} \bot$. Also, by the classical atom clause, since  $\Vdash_{S'} p_c$, we conclude that $p \nvdash_{S'} \bot$, which leads us to a contradiction. Therefore no extension $S'$ of $S$ can satisfy  every $C \in \Gamma_1 \cup \Gamma_2$. Therefore, we conclude vacuously that $\Gamma_1, \Gamma_2 \Vdash_{S} \bot$. Since $S$ is arbitrary, we conclude that $\Gamma_1, \Gamma_2 \Vdash\bot$ by Clause~\ref{Val}.
    \item[$\vee_c$-int]
        Assume, by  inductive hypothesis,  that $\Gamma, \neg A, \neg B \Vdash \bot$. Then by Clause~\ref{Inf} $\Gamma, \neg A, \neg B \Vdash_S \bot$ holds for any arbitrary $S$. Now, take any $S'\supseteq S$ such that  $\Vdash_{S'} C$ for every $C \in \Gamma$. Let $S''\supseteq S'$ arbitrary such that $ A \Vdash_{S''} \bot$ and $B \Vdash_{S''} \bot$. Clause~\ref{toi} yields $\Vdash_{S''} \neg A$ and $\Vdash_{S''} \neg B$. Moreover, monotonicity implies that $\Vdash_{S''} C$ for every $C \in \Gamma$. Hence $\Vdash_{S''} \bot$ by Clause~\ref{Inf}. Since $S''$ is an arbitrary extension where $A \Vdash_{S''} \bot$ and $B \Vdash_{S''} \bot$ implies $\Vdash_{S''} \bot$, we  conclude, using Clause~\ref{orc}, that $\Vdash_{S'} A \lor_c B$. Since $S',S$ are arbitrary, Clause~\ref{Inf} implies that $\Gamma \Vdash_S A \lor_c B$ and hence  $\Gamma \Vdash A \lor_c B$ by Clause~\ref{Val}. 
    \item[$\lor_c$-elim]
        Assume, by  inductive hypothesis, that $\Gamma_1 \Vdash A \lor_c B$, $\Gamma_2, A \Vdash \bot$ and $\Gamma_3, B \Vdash \bot$. Let $S' \supseteq S$ be arbitrary such that $\Vdash_{S'} C$ for all $C \in \Gamma_1 \cup \Gamma_2 \cup \Gamma_3$. By Clause~\ref{Inf},  $\Gamma_1 \Vdash_S A \lor_c B$, $\Gamma_2, A \Vdash_S \bot$ and $\Gamma_3, B \Vdash_S \bot$, and also $\Vdash_{S'} A\lor_c B$, $A \Vdash_{S'} \bot$ and $B \Vdash_{S'} \bot$. Thus $\Vdash_{S'} \bot$ by Clause~\ref{orc}. Since $S',S$ are arbitrary, Clause~\ref{Inf} implies that $\Gamma_1, \Gamma_2, \Gamma_3 \Vdash_S \bot$ and hence $\Gamma_1, \Gamma_2, \Gamma_3 \Vdash \bot$ by Clause~\ref{Val}.
    \item[$\to_c$-int]
        Assume, by  inductive hypothesis, that $\Gamma, A, \neg B \Vdash \bot$. By Clause~\ref{Inf}, $\Gamma,A, \neg B \Vdash_S \bot$ for any arbitrary $S$. Take any arbitrary $S' \supseteq S$ such that $\Vdash_{S'} C$ for all $C \in \Gamma$, and any arbitrary $S''\supseteq S'$ such that $\Vdash_{S''} A$ and $B \Vdash_{S''} \bot$. By Clause~\ref{toi}, we conclude  that $\Vdash_{S''} \neg B$. By monotonicity, we have that $\Vdash_{S''} C$. From $\Vdash_{S''} A$, $\Vdash_{S''} \neg B$, $\Vdash_{S''} C$ for every $C \in \Gamma$ and $\Gamma, A, \neg B \Vdash_{S''} \bot$, we conclude that $\Vdash_{S''} \bot$ by Clause~\ref{Inf}. Now, Clause~\ref{toc} yields $\Vdash_{S'} A \to_c B$ and, since $S'$ is an arbitrary extension, Clause~\ref{Inf} implies that $\Gamma \Vdash_{s} A \to_c B$. Finally, since $S$ is an arbitrary extension, we conclude, using Clause~\ref{Val}, that $\Gamma \Vdash A \to_c B$.  
    \item[$\to_c$-elim]
        Assume, by  inductive hypothesis, that $\Gamma_1 \Vdash A \to_c B$, $\Gamma_2 \Vdash A$ and $\Gamma_3, B \Vdash \bot$. Take an arbitrary $S' \supseteq S$ such that $\Vdash_{S'} C$ for all $C \in \Gamma_1 \cup \Gamma_2 \cup \Gamma_3$.
        Clause~\ref{Inf} then infers that $\Gamma_1 \Vdash_S A \to_c B$, $\Gamma_2 \Vdash_S A$ and $\Gamma_3, B \Vdash_S \bot$, as well as $\Vdash_{S'} A\to_c B$, $ \Vdash_{S'} A$ and $B \Vdash_{S'} \bot$. By Clause~\ref{toc},  $\Vdash_{S'} \bot$. Since $S'$ is an arbitrary extension of $S$, by Clause~\ref{Inf}, it follows that $\Gamma_1, \Gamma_2, \Gamma_3 \Vdash_S \bot$. Finally, Since $S$ is an arbitrary extension, we conclude, using Clause~\ref{Val}, that $\Gamma_1, \Gamma_2, \Gamma_3 \Vdash \bot$
    \item[$\bot$] 
        Assume, by  inductive hypothesis, $\Gamma \Vdash \bot$. Take any $S'\supseteq S$ such that $\Vdash_{S'} C$ for all $C \in \Gamma$. Now, $\Vdash_{S'} \bot$ by Clause~\ref{Inf} and $\vdash_{S'} \bot$ by Clause~\ref{Atm}. But, by the consistency requirement, bases cannot derive $\bot$. Therefore, if $\Gamma \Vdash_S \bot$, no valid extension validates $\Gamma$. Hence we vacuously conclude that every extension of $S$ that validates $\Gamma$ also validates $A$ for any arbitrary $A$. Thus, by Clause~\ref{Inf}, $\Gamma \Vdash_S A$. Since $S$ is arbitrary,  $\Gamma \Vdash A$ by Clause~\ref{Val}.         
\end{enumerate}
\end{proof}

\subsection{Completeness}
We will now prove completeness. The idea is to use an adaptation of Sandqvist's completeness proof in~\cite{Sandqvist2015-SANBSF} with some changes in relation to the consistency constraint, presented in \cite{Pimentel2025-PIMAEV}.

\subsubsection{Preliminary definitions}
The first step is to set an extended notion of subformulas in order to consider the negated discharged formulas in the rules for the classical connectives.
\begin{definition}\label{def:subformulaprop}
{\em Subformulas} are defined as follows.
    \begin{enumerate}
        \item $A$ is a subformula of $A$;
        \item $p_i$ and $\neg p_i$ are subformulas of $p_c$;
        \item $A$ is a subformula of $\neg A$;
        \item All subformulas of $A$ and $B$ are subformulas of $A \land B$, $A \lor_i B$, $A  \to_i B$;
        \item All subformulas of $A$, $B$, $\neg A$ and $\neg B$ are subformulas of $A \lor_{c} B$;
         \item All subformulas of  $A$, $B$ and $\neg B$ are subformulas of $A \to_{c} B$.
    \end{enumerate}
\end{definition}

We then proceed by producing a mapping $\alpha$ that assigns to each formula $A$ a unique $p^A\in\At$.

\begin{definition}
   Let $\Gamma^{*}$ be the set of all subformulas of formulas in $\Gamma$. An {\em atomic mapping} is an injective function $\alpha : \Gamma^{*} \mapsto \At$ such that $\alpha (\bot) = \bot$ and $\alpha (p_i) = p$ for all $p_i \in \Gamma^{*}$. For any other formula $A \in \Gamma^*$, we denote by $p^{A}$ the unique atom to which  $A$ is assigned.
\end{definition}

Notice that, since $\Gamma$ is finite by definition and there are infinitely many atoms in the language, so there are always enough atoms to supply such a mapping. Following Sandqvist's strategy, we start by using the mapping $\alpha$ to construct an atomic system $\mathcal{N}$.

\begin{definition}[Simulation base] 
The rules of the {\em simulation base} $\mathcal{N}$ are depicted in Figure~\ref{Fig:NEp-sim}.
\end{definition}
\begin{figure}[hp] 
{\sc Intuitionistic and neutral rules}
\begin{prooftree}
\AxiomC{$[p^A]$}
\noLine
\UnaryInfC{$\Pi$}
\noLine
\UnaryInfC{$p^B$}
\RightLabel{\scriptsize{$\to_i$-int}}
\UnaryInfC{$p^{A \to_i B}$}
\DisplayProof
\qquad
\AxiomC{$p^{A \to_i B}$}
\AxiomC{$p^A$}
\RightLabel{\scriptsize{$\to_i$-elim}}
\BinaryInfC{$p^B$}
\end{prooftree}

\begin{prooftree}
\AxiomC{$p^{A_j}$}
\RightLabel{\scriptsize{$\lor_i$-int}}
\UnaryInfC{$p^{A_1\lor_i A_2}$}
\DisplayProof
\qquad
\AxiomC{$p^{A \lor_i B}$}
\AxiomC{$[p^A]$}
\noLine
\UnaryInfC{$\Pi_1$}
\noLine
\UnaryInfC{$q$}
\AxiomC{$[p^B]$}
\noLine
\UnaryInfC{$\Pi_2$}
\noLine
\UnaryInfC{$q$}
\RightLabel{\scriptsize{$\lor_i$-elim}}
\TrinaryInfC{$q$}
\end{prooftree}

\begin{prooftree}
\AxiomC{$p^A$}
\AxiomC{$p^B$}
\RightLabel{\scriptsize{$\wedge$-int}}
\BinaryInfC{$p^{A \wedge B}$}
\DisplayProof
\qquad
\AxiomC{$p^{A_1\wedge A_2}$}
\RightLabel{\scriptsize{$\wedge$$-elim_j$}}
\UnaryInfC{$p^{A_j}$}
\end{prooftree}
\begin{prooftree}
\AxiomC{$\bot$}
\RightLabel{\scriptsize{$\bot$-elim}}
\UnaryInfC{$q$}
\DisplayProof
\qquad
\AxiomC{$[p^A]$}
\noLine
\UnaryInfC{$\Pi$}
\noLine
\UnaryInfC{$\bot$}  \RightLabel{\scriptsize{$\neg$-int}}
\UnaryInfC{$p^{\neg A}$}
\DisplayProof
\qquad            
\AxiomC{$p^A$}
\AxiomC{$p^{\neg A}$}   \RightLabel{\scriptsize{$\neg$-elim}}
\BinaryInfC{$\bot$}
\end{prooftree}
{\sc Classical rules}
      \begin{prooftree}
            \AxiomC{$[p^{\neg p_i]}$}
            \noLine
            \UnaryInfC{$\Pi$}
            \noLine
            \UnaryInfC{$\bot$} 
            \RightLabel{\scriptsize{$p_c$-int}} 
            \UnaryInfC{$p^{p_c}$}
\DisplayProof
\qquad            
            \AxiomC{$p^{p_c}$}
            \AxiomC{$p^{\neg p_i}$}
            \RightLabel{\scriptsize{$p_c$-elim}}
            \BinaryInfC{$\bot$}         
        \end{prooftree}
\begin{prooftree}
\AxiomC{$[p^{\neg A}]$}
\AxiomC{$[p^{\neg B}]$}
\noLine
\BinaryInfC{$\Pi$}
\noLine
\UnaryInfC{$\bot$} 
\RightLabel{\scriptsize{$\lor_c$-int}} 
\UnaryInfC{$p^{A \lor_c B}$}
\DisplayProof
\qquad            
\AxiomC{$p^{A \lor_c B}$}
\AxiomC{$p^{\neg A}$}
\AxiomC{$p^{\neg B}$}   \RightLabel{\scriptsize{$\lor_c$-elim}}
\TrinaryInfC{$\bot$}         
\end{prooftree}
\begin{prooftree}
\AxiomC{$[p^A]$}
\AxiomC{$[p^{\neg B}]$}
\noLine
\BinaryInfC{$\Pi$}
\noLine
\UnaryInfC{$\bot$} 
\RightLabel{\scriptsize{$\to_c$-int}} 
\UnaryInfC{$p^{A \to_c B}$}
\DisplayProof
\qquad            
\AxiomC{$p^{A \to_c B}$}
\AxiomC{$p^A$}
\AxiomC{$p^{\neg B}$}            
\RightLabel{\scriptsize{$\to_c$-elim}}
\TrinaryInfC{$\bot$}
\end{prooftree}
\caption{Simulation Base for the Ecumenical natural deduction system $\NEp$.}\label{Fig:NEp-sim}
\end{figure}

We now add atomic rules to $\mathcal{N}$ for all formulas $D \in \Gamma^*$, according to the following criteria:

\begin{enumerate}
    \item For every formula D with shape $A \to_i B,A \land B,\neg A, p_c,A \to_c B$ or $A \lor_c B$, we add $p^D$-int and $p^D$-elim to $\mathcal{N}$;
    \item For every formula $D$ with shape $A \lor_i B$ we add the rules $p^D$-int to $\mathcal{N}$, and for every $D$ with shape $A \lor_i B$ and every $q \in \At$ we add $p^D$-elim with minor premises and conclusion $q$ to $\mathcal{N}$; 
    \item For every formula $D$ with shape  we add $p^D$-int and $p^D$-elim to $\mathcal{N}$;
    
    \item For every $q \in \At$ we add $\bot$-elim to $\mathcal{N}$; 
    \item $\mathcal{N}$ contains no rules other than those added by this procedure.
\end{enumerate}

Since all bases in this section are required to be consistent,  we must show that this is also the case for $\mathcal{N}$. In~\cite{Pereira2017-PERNSA-2} it has been proved that $\mathcal{N}$ is (weak) normalising, hence consistency follows directly from this result.

\begin{theorem}    
$\bot$ is not derivable in $\mathcal{N}$.  \label{consistency}
\end{theorem}

\subsubsection{Completeness proof}
The following result is an adaptation of Lemma~\ref{Lemma 1} to the simulation base case.
\begin{lemma}
 $\forall S\supseteq\mathcal{N}$,   $\Gamma_{\At} \Vdash_S p$ iff $\Gamma_{\At} \vdash_S p$.  \label{Atomic Lemma}
\end{lemma}
\begin{proof} Assume $\Gamma_{\At}\Vdash_S p$. By Lemma~\ref{Lemma 1} we have either $\Gamma_{\At} \vdash_S p$ or $\Gamma_{\At} \vdash_S \bot$. If $\Gamma_{\At} \vdash_S p$ then the proof is already finished. If $\Gamma_{\At} \vdash_S \bot$, then we can apply rule $\bot$-elim at the end of the deduction to obtain $\Gamma_{\At} \vdash_S p$. The converse follows directly from Lemma~\ref{Lemma 1}.
\end{proof}

\begin{lemma}\label{lemma:completenessmain}
   Given any finite $\Gamma$, let $D \in \Gamma^*$, $\alpha$ be an atomic mapping for $\Gamma^*$ and $\mathcal{N}$ a simulation base defined over $D$. Then, $\forall S \supseteq \mathcal{N}$, $\Vdash_{S} D$ iff $\vdash_{S}p^{D}$.
\end{lemma}
\begin{proof} The proof is by structural induction on $D$.

\textbf{Base case.} If $D=p_i$, the result follows directly from Clause~\ref{Atm}.

\textbf{Inductive case.} Consider $D$ to be
\begin{enumerate}
    \item[$p_c$.]
$(\Rightarrow)$
Assume that $\Vdash_S p^c$. Then, for every extension $S'\supseteq S$,
$p\nvdash_{S'}\bot$.
Recall that, by definition, $p^{p_i}=p$. 

Suppose that $p^{\neg p_i}\nvdash_S\bot$. By Lemma~\ref{pc lemma}, $S'= S\cup\{\text{axiom }p^{\neg p_i}\}$ is consistent. But the
simulated negation-elimination rule gives $p\vdash_{S'} \bot$, contradicting $\Vdash_S p^c$. Hence $p^{\neg p_i}\vdash_S \bot$ by classical
reductio, and simulated $p_c$-introduction yields $\vdash_S p^{p_c}$.

$(\Leftarrow)$
Assume that $\vdash_Sp^{p_c}$. Suppose, by contradiction, that
there is an extension $S'\supseteq S$ such that $p\vdash_{S'}\bot$.
Since $p=p^{p_i}$, the $p^{\neg A}$-int rule yields $\vdash_{S'}p^{\neg p_i}.$
Moreover, since $S'\supseteq S$, monotonicity of derivability gives
$ \vdash_{S'}p^{p_c}.$
Applying the $p_c$-elim rule, we obtain
$\vdash_{S'}p^\bot$, and hence $\vdash_{S'}\bot$, contradicting the
consistency of $S'$. Consequently, $p\nvdash_{S'}\bot$ for every  $S'\supseteq S$. 
By Clause~\ref{Atmc}, we conclude that $\Vdash_Sp^c$. 
\item[$A \lor_c B$.]
        ($\Rightarrow$) Assume that $\Vdash_S A \lor_c B$. Now, take any extension $S'$ of $S$ such that $ \Vdash_{S'} p^{\neg A}$ and $ \Vdash_{S'} p^{\neg B}$, as well as an extension $S'' \supseteq S'$ such that $\Vdash_{S''} A$. By the inductive hypothesis, we conclude $\vdash_{S''} p^A$. Since we have $\Vdash_{S'} p^{\neg A}$ we also have $\Vdash_{S''} p^{\neg A}$. Using Clause~\ref{Atm} we conclude $\vdash_{S''} p^{\neg A}$. Since we have $\vdash_{S''} p^{\neg A}$ and $\vdash_{S''} p^A$ we can use the $\neg$-elim to conclude $\vdash_{S''} \bot$. Clause~\ref{Atm} yields $\Vdash_{S''} \bot$. Thus, using Clause~\ref{Inf}, we conclude $A \Vdash_{S'} \bot$. A similar argument shows that $B \Vdash_{S'} \bot$. 
        
        Since we assumed $\Vdash_S A \lor_c B$ and we have $A \Vdash_{S'} \bot$ and $B \Vdash_{S'} \bot$, by  Clause~\ref{orc} we conclude $\Vdash_{S'} \bot$. Now, Clause~\ref{Inf} can be used to obtain $p^{\neg A}, p^{\neg B} \Vdash_S \bot$. Using Lemma \ref{Atomic Lemma}, we obtain $p^{\neg A}, p^{\neg B} \vdash_S \bot$. Finally, we  use the $\lor_{c}$-int rule to conclude $\vdash_Sp^{A \lor_c B}$. 
        
        ($\Leftarrow$) Assume that $\vdash_S p^{A\lor_c B}$. Take any extension $S' \supseteq S$ such that $A \Vdash_{S'} \bot$ and $B \Vdash_{S'} \bot$, and any extension $S'' \supseteq S'$ such that $\Vdash_{S''} p^A$. Then, Clause~\ref{Atm} yields $\vdash_{S''} p^A$. By the inductive hypothesis we obtain $\Vdash_{S''} A$. Since $A \Vdash_{S'} \bot$ and $S'' \supseteq S'$, we conclude $A \Vdash_{S''} \bot$. Using Clause~\ref{Inf} we obtain $\Vdash_{S''} \bot$. Since $S''$ is an arbitrary extension of $S'$ where $\Vdash_{S''} p^A$ holds, we can use Clause~\ref{Inf} to obtain $p^A \Vdash_{S'} \bot$. A similar argument shows that $p^B \Vdash_{S'} \bot$. Using Lemma \ref{Atomic Lemma} we obtain $p^A \vdash_{S'} \bot$ and $p^B \vdash_{S'} \bot$. Using the $\neg$-int rule we conclude $ \vdash_{S'} p^{\neg A}$ and $\vdash_{S'} p^{\neg B}$. Since $S' \supseteq S$, we also have that $\vdash_{S'} p^{A\lor_c B}$. Now, from $ \vdash_{S'} p^{\neg A}$, $\vdash_{S'} p^{\neg B}$ and $\vdash_{S'} p^{A\lor_c B}$, we can use the $\lor_{c}$-elim to obtain a deduction showing that $\vdash_{S'} \bot$. Clause~\ref{Atm}, yields $\Vdash_{S'} \bot$. Therefore, since $S'$ is an arbitrary extension containing $A \Vdash_{S'} \bot$ and $B \Vdash_{S'} \bot$, we conclude that, for every $S' \supseteq S$, $A \Vdash_{S'} \bot$ and $B \Vdash_{S'} \bot$ implies $\Vdash_{S'} \bot$. Finally, Clause~\ref{orc} gives $\Vdash_S A \lor_c B$.   
    \item[$A \to_c B$]
        ($\Rightarrow$) Assume that $\Vdash_S A \to_c B$. Take any extension $S'\supseteq S$ such that $ \Vdash_{S'} p^{A}$ and $ \Vdash_{S'} p^{\neg B}$, and any extension $S'' \supseteq S'$ such that $\Vdash_{S''} B$. Then, using the inductive hypothesis, we conclude $\vdash_{S''} p^B$. Since we have $\Vdash_{S'} p^{\neg B}$ we also have $\Vdash_{S''} p^{\neg B}$. Using Clause~\ref{Atm} we conclude $\vdash_{S''} p^{\neg B}$. Since we have $\vdash_{S''} p^{\neg B}$ and $\vdash_{S''} p^B$ we can use the $\neg$-elim to conclude $\vdash_{S''} \bot$. Clause~\ref{Atm} yields $\Vdash_{S''} \bot$. Thus, using Clause~\ref{Inf}, we conclude $B \Vdash_{S'} \bot$. Using Clause~\ref{Atm}, we conclude $\vdash_{S'} p^A$ from $ \Vdash_{S'} p^{A}$. Using the inductive hypothesis again, we conclude $\Vdash_{S'} A$. Since we assumed $\Vdash_S A \to_c B$ and we have $\Vdash_{S'} A$ and $B \Vdash_{S'} \bot$, we can use Clause~\ref{toc} to conclude $\Vdash_{S'} \bot$. Now, Clause~\ref{Inf} can be used to obtain $p^{A}, p^{\neg B} \Vdash_S \bot$. Using Lemma \ref{Atomic Lemma}, we obtain $p^{A}, p^{\neg B} \vdash_S \bot$. Finally, we can use the $\to_{c}$-int rule to conclude $\vdash_Sp^{A \to_c B}$.  
        
($\Leftarrow$) Assume $\vdash_S p^{A\to_c B}$. Take any extension $S' \supseteq S$ such that $ \Vdash_{S'} A$ and $B \Vdash_{S'} \bot$, and any extension $S'' \supseteq S'$ Such that $\Vdash_{S''} p^B$. Then, Clause~\ref{Atm} yields $\vdash_{S''} p^B$. Using the inductive hypothesis we obtain $\Vdash_{S''} B$. Since $B \Vdash_{S'} \bot$ and $S'' \supseteq S '$, we conclude $B \Vdash_{S''} \bot$. Using Clause~\ref{Inf} we obtain $\Vdash_{S''} \bot$. Since $S''$ is an arbitrary extension of $S'$ where $\Vdash_{S''} p^B$ holds, we can use Clause~\ref{Inf} to obtain $p^B \Vdash_{S'} \bot$. By Lemma \ref{Atomic Lemma} we obtain $p^B \vdash_{S'} \bot$. Using the $\neg$-int rule we conclude $\vdash_{S'} p^{\neg B}$. Since $\Vdash_{S'} A$, the inductive hypothesis yields $\vdash_{S'} p^A$. Since $S' \supseteq S$, we also have that $\vdash_{S'} p^{A\to_c B}$. Now, from $ \vdash_{S'} p^{A}$, $\vdash_{S'} p^{\neg B}$ and $\vdash_{S'} p^{A\to_c B}$, we can use the $\to_{c}$-elim rule to obtain a deduction of $\vdash_{S'} \bot$. Clause~\ref{Atm} yields $\Vdash_{S'} \bot$. Therefore, Since $S'$ is an arbitrary extension containing $\Vdash_{S'} A$ and $B \Vdash_{S'} \bot$, we conclude that for every $S' \supseteq S$, $ \Vdash_{S'} A$ and $B \Vdash_{S'} \bot$ implies $\Vdash_{S'} \bot$. Finally, by Clause~\ref{toc}, we conclude $\Vdash_S A \to_c B$.   
\end{enumerate}
\end{proof}
\begin{theorem}[Completeness] If $\Gamma \Vdash A$ then $\Gamma \vdash_{\NEp} A$.
\end{theorem} 

\begin{proof} Let $\alpha$ and the atomic system $\mathcal{N}$ associated with $\Gamma$ and $A$ be defined as above, and set $\Gamma^{\At} = \{p^A | A \in \Gamma\}$.

Assume that $\Gamma \Vdash A$. By Clause~\ref{Inf} we have $\Gamma \Vdash_{\mathcal{N}} A$. Let $S$ be the system obtained from $\mathcal{N}$ by adding, for every $p^C\in\Gamma^{\At}$, an atomic axiom with conclusion $p^C$.
We distinguish two cases according to whether $S$ is consistent.

\medskip
\noindent
\textbf{Case 1: $S$ is consistent.}
Then $S$ is an admissible extension of $\mathcal{N}$. By construction, $\vdash_S p^C$ for every $C\in\Gamma$.
Since $\Gamma^*\supseteq\Gamma$, Lemma~\ref{lemma:completenessmain}
gives $\Vdash_SC$ for every $C\in\Gamma$.
As $\Gamma\Vdash_{\mathcal{N}}A$ and $S\supseteq\mathcal{N}$, Clause~\ref{Inf} yields $\Vdash_SA$. Applying Lemma~\ref{lemma:completenessmain} once more, we obtain $\vdash_Sp^A$.

Let $\Pi$ be a deduction witnessing $\vdash_S p^A$. Replace every application in $\Pi$ of a newly added axiom with conclusion $p^C$ by an open assumption $p^C$. The resulting deduction uses only rules from $\mathcal{N}$ and therefore shows $\Delta\vdash_{\mathcal{N}}p^A$ for some $\Delta\subseteq\Gamma^{\At}$. By weakening, $ \Gamma^{\At}\vdash_{\mathcal{N}}p^A$.

We now translate this deduction back into the language of $\NEp$, replacing every occurrence of $p^C$ in the image of $\alpha$ by $C$. The auxiliary atoms occurring in instances of $\lor_i$-elim and $\bot$-elim, which are not in the image of $\alpha$, are left unchanged. By the construction of $\mathcal{N}$, every atomic rule is thereby transformed into an instance of a rule of $\NEp$.
An induction on the height of the deduction consequently yields $\Gamma\vdash_{\NEp}A$.

\medskip
\noindent
\textbf{Case 2: $S$ is inconsistent.}
In this case, there is a deduction $\Pi$ witnessing $\vdash_S\bot$.  Replace every application of a newly added axiom with conclusion $p^C$ by an open assumption $p^C$. This gives a deduction using only rules from $\mathcal{N}$ and showing $\Delta\vdash_{\mathcal{N}}\bot$  for some $\Delta\subseteq\Gamma^{\At}$. Notice that $\Delta$ must be non-empty: otherwise, we would have $\vdash_{\mathcal{N}}\bot$, contradicting the consistency of $\mathcal{N}$ established in Theorem~\ref{consistency}.

Applying the same translation as in the first case and arguing by induction on the height of the deduction, we obtain $\Gamma\vdash_{\NEp}\bot$. 
Finally, by $\bot$-elim, $\Gamma\vdash_{\NEp}A$.

Thus, in either case, $\Gamma\vdash_{\NEp}A$.
\end{proof}


\section{Connecting Ecumenical $\Bes$ and Beth semantics}\label{sec:Beth}

Our presentation of Beth semantics mainly follows \cite{BeerBeth}, with a few elements taken from \cite{sep-logic-intuitionistic}. Readers interested in notationally different formulations may consult \cite{TroelstraConstructivismVol2,kripke1965semantical}. For convenience, we start by restating the technical definition of atomic bases as follows, following the convention in \cite{Followingall10.1093/analys/anac100}.

\begin{definition} We will denote by $\mathscr{R}$ the set of all atomic rules, and by $\mathcal{P}(\mathscr{R})$ the set of all bases, were $\mathcal{P}$ is the usual powerset operator. The set 
of all consistent bases is defined as $\mathfrak{C}(\mathscr{R}) = \{S \in \mathcal{P(\mathscr{R})} \ | \ \nvdash_{S} \bot \}$
\end{definition}

\begin{definition}
The support relation $\Vdash^{\mathfrak{C}}_{S}$ is defined by the clauses for $p^{i}$, $\land$, $\lor_i$, $\to_i$, and semantic consequence (Clause~\ref{Inf}) in Definition~\ref{def:validity}, with all bases and base extensions ranging over $\mathfrak{C}(\mathscr{R})$.
The support relation $\Vdash^{\mathfrak{R}}_{S}$ is defined by the same
clauses, with all bases and base extensions ranging over
$\mathcal{P}(\mathscr{R})$.
\end{definition}

Both classes of bases are uncountable. Accordingly, the definitions
of frames and paths below are formulated without countability
restrictions. We include only the intuitionistic and neutral clauses, since these suffice to develop the comparison with
Beth semantics.

In both settings, $\bot$ belongs to $\At$ and may occur in atomic
rules, with its support determined by atomic derivability.
Its distinguished role in the consistent setting arises from the
requirement that $\nvdash_S \bot$ for every admissible base $S$.
In the unrestricted setting, this requirement is dropped and no
separate semantic clause for absurdity is imposed. Consequently,
$\bot$ behaves like any other basic sentence.

\begin{definition}
An \emph{ordered frame} is a pair $F=\langle W,\leq\rangle$,
where $W$ is a set and $\leq$ is a reflexive and transitive
relation on $W$.

For a base $S$, let
\[
\begin{aligned}
\pi(S) &= \{S' \in \mathcal{P}(\mathscr{R}) \mid S' \supseteq S\},\\
\eta(S) &= \{S' \in \mathfrak{C}(\mathscr{R}) \mid S' \supseteq S\}.
\end{aligned}
\]
The \emph{unrestricted extension frame} and the
\emph{consistent extension frame} of $S$ are defined, respectively, by
\[
F(S)=\langle\pi(S),\subseteq\rangle
\qquad\text{and}\qquad
C(S)=\langle\eta(S),\subseteq\rangle.
\]
\end{definition}

Hence $F(S)$ is the set of extensions of $S$ ordered by $\subseteq$, and $C(S)$ the ordered set of its consistent extensions. Note that $F(\varnothing) = \langle \mathcal{P}(\mathscr{R}), \subseteq \rangle$ and  $C(\varnothing) = \langle\mathfrak{C}(\mathscr{R}), \subseteq \rangle$.

The definition of Beth semantics strongly relies on the concept of \textit{path}.

\begin{definition}\label{def:maxchain}
A \emph{chain} in an ordered frame $\langle W,\leq\rangle$
is a subset $\alpha\subseteq W$ such that, for all $w,v\in\alpha$,
either $w\leq v$ or $v\leq w$.

A \emph{path}, or \emph{maximal chain}, in
$\langle W,\leq\rangle$ is a chain $\alpha$ that is maximal
under inclusion: whenever $\alpha\subseteq\beta$
and $\beta$ is a chain, we have $\alpha=\beta$.
\end{definition}

The following result is standard see \eg,~\cite[Theorem~2.4.1]{GoubaultLarrecq2013}.

\begin{lemma}[Hausdorff's maximal principle]
Any chain in a partially ordered set can be extended to a maximal chain.
\end{lemma}
It folows directly from Hausdorff's maximal principle that any two comparable elements of an
ordered frame belong to a common path: if $w\leq w'$, then
$\{w,w'\}$ is a chain and can therefore be extended to a maximal chain.

On tree frames, paths coincide with maximal branches, providing a
familiar picture of alternative developments. The frames induced by
base extension need not be trees, but the maximal-chain definition
applies directly to them without imposing any countability restriction.

As in other variants of proof-theoretic semantics
\cite[p.~200]{Goldfarb2015-GOLODP-2}
\cite{Piecha2015FailComplete-PIEFOC}, the support clauses of $\Bes$
can be understood as implicitly defining a semantics over an ordered
frame. In our setting, the underlying frame is $F(\varnothing)$
for unrestricted bases and $C(\varnothing)$ for consistent bases.

For the support clauses introduced above, evaluating support at an
admissible base $S$ in the global frame is equivalent to evaluating
it in $F(S)$ or $C(S)$, respectively. Indeed, these restricted frames
contain exactly the admissible extensions of $S$, and each such
extension has the same admissible extensions in either presentation.
The frame is thus determined by the chosen class of bases and the
extension relation, together with $S$ when the restricted presentation
is used. It may therefore be left implicit in the notation for support.

To compare $\Bes$ with Beth semantics, we now make these frames
explicit and introduce semantic clauses that quantify over their
paths. For the following definitions, we work with the global frames
$F(\varnothing)$ and $C(\varnothing)$. The mixed semantic clauses are
defined as follows.
\begin{definition}
    $\Bes$-Beth mixed semantic clauses are defined as follows:

    \begin{enumerate}
        \item $\Vdash^{\mathfrak{R}}_{S} p_{b}$ iff, for all paths $\alpha$ through $F(\varnothing)$, if $S \in \alpha$ then there is an $S' \in \alpha$ such that $\vdash_{S'} p$. 

        \item $\Vdash^{\mathfrak{C}}_{S} p_{b}$ iff, for all paths $\alpha$ through $C(\varnothing)$, if $S \in \alpha$ then there is an $S' \in \alpha$ such that $\vdash_{S'} p$. 

         \item $\Vdash^{\mathfrak{R}}_{S} A \lor_{b} B$ iff, for all paths $\alpha$ through $F(\varnothing)$, if $S \in \alpha$ then there is a $S' \in \alpha$ such that either $\Vdash^{\mathfrak{R}}_{S'} A$ or $\Vdash^{\mathfrak{R}}_{S'} B$.

          \item $\Vdash^{\mathfrak{C}}_{S} A \lor_{b} B$ iff, for all paths $\alpha$ through $C(\varnothing)$, if $S \in \alpha$ then there is a $S' \in \alpha$ such that either $\Vdash^{\mathfrak{C}}_{S'} A$ or $\Vdash^{\mathfrak{C}}_{S'} B$.
    \end{enumerate}
\end{definition}

Notice that, in particular, $\vdash_{S} p$ implies $\Vdash^{\mathfrak{R}}_{S} p_{b}$ and $\Vdash^{\mathfrak{C}}_{S} p_{b}$.

Now we borrow some $\Bes$ and model-theoretic lemmata from \cite{Pimentel2025-PIMAEV,Makinson2014-MAKOAI-3,BeerBeth} and prove some additional set-theoretic results.

\begin{lemma}[Monotonicity]
If $\Gamma \Vdash^{\mathfrak{C}}_{S} A$ and $S'\supseteq S$ then $\Gamma \Vdash^{\mathfrak{C}}_{S'} A$
\end{lemma}
\begin{proof}


  For $\Gamma = \varnothing$, either $A$ is an intuitionistic connective, in which case the result follows from Lemma \ref{monotonicity}, or it is either $p_{b}$ or $A \lor_{b} B$, in which case we reason as follows. For $p_{b}$, assume $\Vdash^{\mathfrak{C}}_{S} p_{b}$ and $S' \supseteq S $. Let $\alpha$ be any path through $C(\varnothing)$ with $S' \in \alpha$ and define $\beta = \{ S'' \in \alpha \ | \ S'' \supseteq S'\}$. Since $\beta$ is a subset of $\alpha$ and $\alpha$ it is totally ordered, $\beta$ is a chain. Since  by transitivity of $\supseteq$ we have $S'' \supseteq S$ for all $S'' \in \beta$, we conclude that $\{ S\} \cup \beta $ is chain, which can be extended to a maximal chain $\theta$. Notice that, if there is a $S^{1} \in \theta$ with $S^{1} \supseteq S'$ then  $S^{1} \in \beta$; otherwise, since $S' \supseteq S^{2}$ for all $S^{2} \in (\alpha- \beta)$ and also  $S^{1} \supseteq S'$, by transitivity we conclude $S^{1} \supseteq S^{2}$ for all $S^{2} \in (\alpha - \beta)$ and, since we have either $S^{1} \supseteq S^{3}$ or $S^{3} \supseteq S^{1}$ for all $S^{3} \in \beta$ (as $\beta \subseteq \theta$, $S^{3} \in \theta$ and $\theta$ is a maximal chain), $S^{1} \cup \alpha$ is totally ordered, whence from maximality of $\alpha$ we get $S^{1} \in \alpha$ and, by definition of $\beta$, $S^{1} \in \beta$. Now, since $\Vdash^{\mathfrak{C}}_{S} p_{b}$ and $\theta$ is a maximal chain with $S \in \theta$, there must be a $S^{4} \in \theta$ such that $\vdash_{S^{4}} p$. If $S^{4} \supseteq S'$ then $S^{4} \in \beta$ and so $S^{4} \in \alpha$, and if $S' \supseteq S^{4}$ then $\vdash_{S^{'}} p$, so since $S' \in \alpha$ in any case we conclude that there is a $S^{5} \in \alpha$ such that $\vdash_{S^{5}} p$ (either an extension of $S'$ in $\beta$ or $S'$ itself), hence by arbitrariness of $\alpha$ we conclude $\Vdash^{\mathfrak{C}}_{S'} p_{b}$. The proof for $A \lor_{b} B$ is similar, the only difference being that we conclude that there must be a $S^{4}$ with either $\Vdash^{\mathfrak{C}}_{S^{4}} A$ or  $\Vdash^{\mathfrak{C}}_{S^{4}} B$ and then conclude that either $\Vdash^{\mathfrak{C}}_{S'} A$ or  $\Vdash^{\mathfrak{C}}_{S'} B$ by the inductive hypothesis. For $\Gamma \neq \varnothing$ the proof is also identical to that of Lemma \ref{monotonicity}.
 \end{proof}

\begin{lemma}\label{lemma:maximalconsistencyexistence}
Every consistent system $S$ has at least one maximal consistent extension $S^{M}$ such that:

\begin{enumerate}
    \item For all $S'$, if $S' \supseteq S^{M}$ and $S'$ is consistent then $S^{M} = S'$.

      \item For all $p$, either $\vdash_{S^{M}} p$ or $p \vdash_{S^{M}} \bot$.
\end{enumerate}
\end{lemma}
\begin{proof}
  The proof is obtained by combining Lemma 14 of \cite{Pimentel2025-PIMAEV} with Lemma 3.3 of \cite{Makinson2014-MAKOAI-3}. We start by defining an enumeration of atomic rules and using the same extension procedure as in Lemma 14 of \cite{Pimentel2025-PIMAEV}. Instead of either adding a rule concluding a single atom $p$ at every step $S^{m}$ if the resulting base is consistent and letting it unchanged otherwise, we add the rule assigned to the number $m$ if the resulting base is consistent and let the base unchanged otherwise. The proof that used to show that the union of all $S^{m }$ with $m \geq 0$ is consistent can then be repeated, and to show that $S' \supseteq S^{M}$ implies $S' = S^{M}$ we only have to point out that, due to the definition of the $S^{m}$ and their union, every rule that can be added to $S^{M}$ whilst preserving consistency is already in $S^{M}$. The proof that either $\vdash_{S^{M}} p$ or $p\vdash_{S^{M}} \bot$ remain exactly the same.
\end{proof}

\begin{corollary}\label{lemma:maximalconsistentpersistency}
Let $S^{M}$ be a maximal consistent extension of $S$. Then, for every $S' \supseteq S^{M}$ and every $A$, $\Vdash^{\mathfrak{C}}_{S^{M}} A$ iff $\Vdash^{\mathfrak{C}}_{S'} A$.
\end{corollary}

\begin{lemma}\label{lemma:containssetofallrules}
    Every path $\alpha$ through $F(\varnothing)$ contains $\mathscr{R}$.
\end{lemma}
\begin{proof}
The result follows immediately from the facts that $\mathscr{R} \supseteq S$ for any $S$ and that $\alpha$ is a maximal chain.
\end{proof}

\begin{lemma}\label{lemma:existenceunionofconsistentpathisconsistent}
Let $S$ be a base. Let $\cup \alpha$ be the union of all sets in a path $\alpha$ through $C(\varnothing)$ or $F(\varnothing)$. Then if $ \cup \alpha \supseteq S$ there is at least one $S' \in \alpha$ such that $S' \supseteq S$.
\end{lemma}
\begin{proof}
Since $\cup \alpha \supseteq S$, for any rule $R^{i} \in S$ there must be at least one $S^{i} \in \alpha$ with $R^{i} \in S^{i}$. Take any arbitrary $S^{i}$ for every $R^{i}$, and let $\Theta$ be the set of the chosen $S^{i}$. Since $\alpha$ is totally ordered and $S^{i} \in \alpha$ for all $S^{i}$, for any pair of elements $S^{j}, S^{k}$ in $\Theta$  we have either $S^{j} \supseteq S^{k}$, in which case $S^{j} \cup S^{k} = S^{k}$, or $S^{k} \supseteq S^{j}$, in which case $S^{j} \cup S^{k} = S^{j}$. From this it follows that the union $S'$ of all elements of $\Theta$ is an element of $\Theta$, hence since all elements of $\Theta$ are elements of $\alpha$ we have  $S' \in \alpha$. Since for every $R^{i} \in S$ there is a $S^{i} \in \Theta$ such that $R^{i} \in S^{i}$ we also have that $R^{i} \in S'$ for all $R^{i} \in S$, so we conclude $S' \supseteq S$.
\end{proof}

\begin{lemma}\label{lemma:unionofconsistentpathisconsistent}
    If $\alpha$ is a path through $C(\varnothing)$ then the union of all sets in $\alpha$ is consistent.
\end{lemma}
\begin{proof}
Suppose that $\alpha$ is a path through $C(\varnothing)$ but the union of all its elements $\cup \alpha$ is inconsistent. Then we have $\vdash_{\cup \alpha} \bot$. Since proofs are finitary objects, let $\Pi$ be any proof of $\bot$ in $\cup \alpha$ and $S$ be the set of all rules used in $\Pi$. Clearly, $\cup \alpha \supseteq S$. By Lemma \ref{lemma:existenceunionofconsistentpathisconsistent}, there is a $S' \in \alpha$ such that $S' \supseteq S$. But then $S'$ is inconsistent and, since all elements of any path through $C(\varnothing)$ are consistent by definition, contradiction. Therefore, $\cup\alpha$ must be consistent.
\end{proof}

\begin{lemma}\label{lemma:everypaththroughconsistentcontainsmaxiconsistent}
    Every path $\alpha$ through $C(\varnothing)$ contains a maximal consistent base.
\end{lemma}
\begin{proof}
       Let $\alpha$ be a path through $C(\varnothing)$ and $\cup \alpha$ the union of all its elements.  By Lemma \ref{lemma:unionofconsistentpathisconsistent}, $\cup \alpha$ is consistent. Clearly, $\cup \alpha \supseteq S$ for all $S \in \alpha$. By Lemma \ref{lemma:maximalconsistencyexistence}, there is a maximal consistent base $S^{M} \supseteq \cup \alpha$, whence by transitivity of $\supseteq$ we have $S^{M} \supseteq  S$ for all $S \in \alpha$, hence since $\alpha$ is a maximal chain we have $S^{M} \in \alpha$.
\end{proof}



Now we proceed to our main results.
\begin{theorem} The following hold:

\begin{enumerate}
    \item     $\Vdash^{\mathfrak{R}}_{S} p_{b}$ holds for all $p \in \At$ in any $S$;

      \item     $\Vdash^{\mathfrak{R}}_{S} A \lor_{b} B$ holds for any $A$ and $B$ in any $S$;

\end{enumerate}

\end{theorem}
\begin{proof}
Let $\alpha$ be any path through $F(\varnothing)$ with $S \in \alpha$. By Lemma \ref{lemma:containssetofallrules}, $\mathscr{R} \in \alpha$. Since $\mathscr{R}$ contains a rule concluding $p$ from no premises for every atom $p$, we have $\vdash_{\mathscr{R}} p$ for all $p$, hence $\Vdash^{\mathfrak{R}}_{S} p_{b}$ for all $p$. It is also straightforward to show through induction on formulas that $\Vdash^{\mathfrak{R}}_{\mathscr{R}} A$ holds for all $A$, hence we also have $\Vdash^{\mathfrak{R}}_{S} A \lor_{b} B$ for any $A$ and $B$.
\end{proof}

\begin{theorem} The following hold:

\begin{enumerate}
    \item     $\Vdash^{\mathfrak{C}}_{S} p_{b}$ iff, for all $S' \supseteq S$, $p \nvdash_{S'} \bot$;

      \item     $\Vdash^{\mathfrak{C}}_{S} A \lor_{b} B$ iff, for all $S' \supseteq S$, if $A \Vdash^{\mathfrak{C}}_{S'} \bot$ and $B \Vdash^{\mathfrak{C}}_{S'} \bot$ then $\Vdash^{\mathfrak{C}}_{S'} \bot$;

\end{enumerate}

\end{theorem}
\begin{proof}\ 
\begin{enumerate}
    \item[(1, $\Rightarrow$)] Assume that $\Vdash^{\mathfrak{C}}_{S} p_{b}$ but that there is some $S' \supseteq S$ with  $p \vdash_{S'} \bot$. Pick any path $\alpha$ with $S \in \alpha$ and $S' \in \alpha$. Since $\Vdash^{\mathfrak{C}}_{S} p_{b}$ we conclude that there is a $S'' \in \alpha$ such that $\vdash_{S''} p$. Since $\alpha$ is a maximal chain we conclude that either $S' \supseteq S''$ or $S'' \supseteq S'$. In the first case, $S'$ is inconsistent, and in the second $S''$ is inconsistent, so in either case we conclude that there is an inconsistent base in $\alpha$ through $C(\varnothing)$. Contradiction. Therefore, for all $S \supseteq S'$ we have $p \nvdash_{S'} \bot$.

     \item[(1, $\Leftarrow$)]  Assume that, for all $S' \supseteq S$, $p \nvdash_{S'} \bot$.
     
     Let $\alpha$ be any path such that $S \in \alpha$.  By Lemma \ref{lemma:everypaththroughconsistentcontainsmaxiconsistent}, there is a $S^{M}$ such that $S^{M} \supseteq S$ and $S^{M} \in \alpha$. Since $S^{M} \supseteq S$ we conclude $p \nvdash_{S^{M}} \bot$. By Lemma \ref{lemma:maximalconsistencyexistence}, either $\vdash_{S^{M}} p$ or $p \vdash_{S} \bot$, and since $p \nvdash_{S^{M}} \bot$ we conclude $\vdash_{S^{M}} p$, whence $\Vdash^{\mathfrak{C}}_{S} p$ by arbitrariness of $\alpha$.


     \item[(2, $\Rightarrow$)] Assume $\Vdash^{\mathfrak{C}}_{S} A \lor_{b} B$. For the sake of proving a contradiction, assume that there is a $S' \supseteq S$ such that $A \Vdash^{\mathfrak{C}}_{S'} \bot$ and $B \Vdash^{\mathfrak{C}}_{S'} \bot$. Pick any $\alpha$ with $S \in \alpha$ and $S' \in \alpha$. Since $\Vdash^{\mathfrak{C}}_{S} A \lor_{b} B$, there must be a $S'' \in \alpha$ such that either $\Vdash^{\mathfrak{C}}_{S''} A$ or $\Vdash^{\mathfrak{C}}_{S''} B$. Now pick the maximal consistent set $S^{M}$ with $S^{M} \in \alpha$. Since $S^{M} \supseteq S'$, by monotonicity we conclude that both $A \Vdash^{\mathfrak{C}}_{S} \bot$ and  $B \Vdash^{\mathfrak{C}}_{S} \bot$ hold. Since $S^{M} \supseteq S''$, by monotonicity we conclude that either  $\Vdash^{\mathfrak{C}}_{S} A$ or $\Vdash^{\mathfrak{C}}_{S} B$ hold, hence in any case $\Vdash^{\mathfrak{C}}_{S^{M}} \bot$. Then $\vdash_{S^{M}} \bot$, contradicting the fact that $S^{M}$ is consistent. Therefore, there cannot be any $S' \supseteq S$ such that $A \Vdash^{\mathfrak{C}}_{S'} \bot$ and $B \Vdash^{\mathfrak{C}}_{S'} \bot$, hence it vacuously holds that, for all $S \supseteq S'$, if  $A \Vdash^{\mathfrak{C}}_{S'} \bot$ and $B \Vdash^{\mathfrak{C}}_{S'} \bot$ then $\Vdash^{\mathfrak{C}}_{S'} \bot$.

      \item[(2, $\Leftarrow$)] Assume that, for all $S \supseteq S'$, if $A \Vdash^{\mathfrak{C}}_{S'} \bot$ and $B \Vdash^{\mathfrak{C}}_{S'} \bot$ then $\Vdash^{\mathfrak{C}}_{S'} \bot$. Take any arbitrary $\alpha$ with $S \in \alpha$. For the sake of proving a contradiction, assume that $\nVdash^{\mathfrak{C}}_{S''} A$ and $\nVdash^{\mathfrak{C}}_{S''} B$ for all $S'' \in \alpha$. Now consider the maximal consistent extension $S^{M}$ with $S^{M} \in \alpha$. Since $\nVdash^{\mathfrak{C}}_{S^{M}} A$ and for all $S''' \supseteq S^{M}$ we have $S^{M} = S'''$ by the first part of Lemma \ref{lemma:maximalconsistencyexistence}, it is vacuously true that $A \Vdash^{\mathfrak{C}}_{S^{M}} \bot$. Likewise, it is vacuously true that $B \Vdash^{\mathfrak{C}}_{S^{M}} \bot$, so since $S^{M} \supseteq S$ our original assumption yields $\Vdash^{\mathfrak{C}}_{S^{M}} \bot$ and then $\vdash_{S^{M}} \bot$, contradicting the fact that $S^{M}$ is consistent. Therefore, there must be some $S'' \in \alpha$ such that either $\Vdash^{\mathfrak{C}}_{S''} A$ or $\Vdash^{\mathfrak{C}}_{S''} B$, whence from arbitrariness of $\alpha$ we conclude $\Vdash^{\mathfrak{C}}_{S} A \lor_{b} B$. 
\end{enumerate}
\end{proof}

This proves that $\Vdash^{\mathfrak{C}}_{S} p_{b}$ and $\Vdash^{\mathfrak{C}}_{S} A \lor_{b} B$ are just alternative definitions of $\Vdash_{S} p_{c}$ and $\Vdash_{S} A \lor_{c} B$, meaning that by combining the semantic clauses of Beth with $\Bes$ defined over consistent bases we get definitions of classical connectives.

As mentioned before, the fact that bases are defined through sets and set-theoretic inclusion allows frames to be left implicit. This structure also allows an additional simplification when Beth clauses are concerned. As mentioned before, validity in a base $S$ can be defined either by considering the frames $F(\varnothing)$ and $C(\varnothing)$ or the frames $F(S)$ and $C(S)$. We briefly prove this fact:

\bigskip

\begin{theorem}
Let $\Vdash^{\mathfrak{C}^{*}}_{S}$ be the relation defined through the same clauses as $\Vdash^{C}_{S}$ but over the frame $C(S)$ instead of $C(\varnothing)$. Then $\Vdash^{\mathfrak{C}}_{S} A$ iff $\Vdash^{\mathfrak{C^{*}}}_{S} A$.
\end{theorem}
\begin{proof}
Since the clauses for intuitionistic connectives make no reference to frames and paths, the proof for them is immediate. For $p_{b}$, we reason as follows: 

Assume $\Vdash^{\mathfrak{C}}_{S} p_{b}$. Pick any maximal chain $\alpha$ through $C(S)$. Since $S$ is the least element of $C(S)$ we have $S \in \alpha$. Given that $C(S) \subseteq C(\varnothing)$, $\alpha$ is a chain in $C(\varnothing)$. Let $\alpha \cup \beta$ be the extension of $\alpha$ to a maximal chain in $C(\varnothing)$. Notice that $\beta$ cannot contain any $S'$ with $S' \supseteq S$ and $S' \notin \alpha$, as since $C(S)$ contains all extensions of $S$ this would imply that $S' \in C(S)$ and $\{S'\} \cup \alpha$ is totally ordered but $S \notin \alpha$, contradicting the fact that $\alpha$ was maximal in $C(S)$. Since $\Vdash^{\mathfrak{C}}_{S} p_{b}$ and $S \in (\alpha \cup \beta)$, there is a $S' \in (\alpha \cup \beta)$ such that $\vdash_{S'} p$. If $S' \notin \alpha$ then $S' \in \beta$, so since $S'$ is not an extension of $S$ and $\alpha \cup \beta$ is maximally ordered we conclude  $S \supseteq S'$, so also $\vdash_{S} p$, whence there is a $S' \in \alpha$ (namely, $S$) with $\vdash_{S'} p$. Otherwise, $S' \in \alpha$, so in either case there is a $S' \in \alpha$ with $\vdash_{S'} p$, so from arbitrariness of $\alpha$ we conclude $\Vdash^{\mathfrak{C}^{*}}_{S} p_{b}$. For the converse, assume $\Vdash^{\mathfrak{C}^{*}}_{S} p_{b}$. Let $\alpha$ be a maximal chain in $C(\varnothing)$ with $S \in \alpha$, and let $\beta = \{ S' \in \alpha \ | \ S' \supseteq S\}$. Clearly, $\beta$ is a maximal chain in $C(S)$, and since $S \in \beta$ and $\Vdash^{\mathfrak{C}^{*}}_{S} p_{b}$ there must be a $S' \in \beta$ such that $\vdash_{S'} p$. Since $\beta \subseteq \alpha$ we conclude that there is a $S' \in \alpha$ such that $\vdash_{S'} p$, so from arbitrariness of $\alpha$ we conclude $\Vdash^{\mathfrak{C}}_{S} p_{b}$. The proof for $A \lor_{b} B$ is analogous, the only difference being that if $S' \in \beta$ and  (without loss of generality) $\Vdash^{\mathfrak{C}}_{S'} A$, we conclude $\Vdash^{\mathfrak{C}}_{S} A$ by monotonicity instead of by proof-theoretic properties of bases. 
\end{proof}

Although definitions using $C(\varnothing)$ have the advantage of dealing with a single structure regardless of the base, if we use $C(S)$ some additional simplifications become available. Since $S$ is the least element of $C(S)$, it is included in all maximal chains through it, so we no longer need to require $S \in \alpha$ for the semantic definitions to work. In other words, the following hold:

\begin{enumerate}
    \item $\Vdash^{\mathfrak{C}^{*}}_{S} p_{b}$  iff, in any maximal chain $\alpha$ through $C(S)$, there is a $S' \in \alpha$ such that $\vdash_{S'} p$;

        \item $\Vdash^{\mathfrak{C}^{*}}_{S} A \lor_{b} B$  iff, in any maximal chain $\alpha$ through $C(S)$, there is a $S' \in \alpha$ such that either $\Vdash_{S'} A$ or $\Vdash_{S'} B$;
\end{enumerate}

Moreover, since maximal chains in $C(S)$ are just sets of extensions of $S$, we can entirely remove the reference to paths and define both clauses solely in terms of the extension relation. Guaranteeing the existence of some $S'$ satisfying a condition in all maximal chains through $C(S)$ is, roughly speaking, the same as guaranteeing that in any extension process starting at $S$ we will eventually reach a $S'$ in which the condition is satisfied, meaning that we can never go to a point where no extension satisfying the condition is available. In other words, we may define classical validity in $S$ by requiring that, for all $S' \supseteq S$, there is at least some $S'' \supseteq S'$ such that a given condition is satisfied.  Similar reasoning can  be used to define a connective $A \to_{b} B$ which is a semantic alternative to $A \to_{c}B$, even though it has no counterpart in Beth semantics. 

The simplified semantic clauses are be defined as follows.
\begin{enumerate}
  \item $\Vdash^{\mathfrak{C}^{*}}_{S} p_{b}$ iff, for all $S' \supseteq S$, there is some $S'' \supseteq S'$ such that $\vdash_{S''} p$ .

    \item $\Vdash^{\mathfrak{C}^{*}}_{S} A \lor_{b} B$ iff, for all $S' \supseteq S$, there is some $S'' \supseteq S'$ such that either $\Vdash^{\mathfrak{C}^{*}}_{S''} A$ or $\Vdash^{\mathfrak{C}^{*}}_{S''} B$

     \item $\Vdash^{\mathfrak{C}^{*}}_{S} A \to_{b} B$ iff, for all $S' \supseteq S$, there is some $S'' \supseteq S'$ such that $A \Vdash^{\mathfrak{C}^{*}}_{S''} B$.
\end{enumerate}

We conclude this section by going back to our original definitions and showing that the clauses above are alternative definitions to the classical operators:

\begin{theorem}
    The following hold:  
\begin{enumerate}
  \item $\Vdash_{S} p^{c}$ iff, for all $S' \supseteq S$, there is some $S'' \supseteq S'$ such that $\vdash_{S''} p$ .

    \item $\Vdash_{S} A \lor_{c} B$ iff, for all $S' \supseteq S$, there is some $S'' \supseteq S'$ such that either $\Vdash_{S''} A$ or $\Vdash_{S''} B$

     \item $\Vdash_{S} A \to_{c} B$ iff, for all $S' \supseteq S$, there is some $S'' \supseteq S'$ such that $A \Vdash_{S''} B$.
\end{enumerate}
\end{theorem}
\begin{proof}\

\begin{enumerate}
    \item[(1, $\Rightarrow$)] Assume $\Vdash_{S} p^{c}$. Then, for all $S' \supseteq S$, $p\nvdash_{S'} \bot$. Take any arbitrary $S' \supseteq S$ By Lemma \ref{lemma:maximalconsistencyexistence}, there is a $S^{M} \supseteq S'$ such that either $\vdash_{S^{M}} p$ or $p \vdash_{S^{M}} \bot$. Since $S^{M} \supseteq S$ we have $p \nvdash_{S^{M}} \bot$, so by Lemma \ref{lemma:maximalconsistencyexistence}  we conclude $\vdash_{S''} p$.

     \item[(1, $\Leftarrow$)] Assume that, for all $S' \supseteq S$, there is some $S'' \supseteq S'$ such that $\vdash_{S''} p$. For the sake of proving a contradiction, assume $\nVdash_{S} p^{c}$. Then there is a $S''' \supseteq S$ with $p \vdash_{S'''} \bot$. By our first assumption, there is a $S'''' \supseteq S'''$ with $\vdash_{S''''} p$,  and since $S'''' \supseteq S'''$ we have $p \vdash_{S''''} \bot$. But then $S''''$ is inconsistent, which contradicts the consistency constraint, so we conclude $\Vdash_{S} p^{c}$.

     \item[(2, $\Rightarrow$)] Assume $\Vdash_{S} A \lor_{c} B$. Then, for all $S' \supseteq S$, if $A \Vdash_{S'} \bot$ and $B \Vdash_{S'} \bot$ then $\Vdash_{S'} \bot$. Assume, for the sake of proving a contradiction, that there is a $S' \supseteq S$ such that, for all $S'' \supseteq S'$, both $\nVdash_{S''} A$ and $\nVdash_{S''} B$. Then it vacuously holds that $A \Vdash_{S'} \bot$ and $B \Vdash_{S'} \bot$, so $\Vdash_{S'} \bot$ and $\vdash_{S'} \bot$. Contradiction. Therefore, for all $S' \supseteq S$ there must be some $S'' \supseteq S'$ such that either $\Vdash_{S''} A$ or $\Vdash_{S''} B$.

      \item[(2, $\Leftarrow$)] Assume that, for all $S' \supseteq S$ there is some $S'' \supseteq S'$ such that either $\Vdash_{S''} A$ or $\Vdash_{S''} B$. Assume, for the sake of proving a contradiciton, that there is a $S' \supseteq S$ such that $A \Vdash_{S'} \bot$ and $B \Vdash_{S'} \bot$. By our first assumption, there is a $S'' \supseteq S'$ such that either $\Vdash_{S''} A$ or $\Vdash_{S''} B$. In either case we conclude $\Vdash_{S''} \bot$ by the entailment clause, yielding a contradiction. Therefore, there cannot be a $S' \supseteq S$ with $A \Vdash_{S'} \bot$ and $B \Vdash_{S'} \bot$, meaning that it is vacuously true that, for all $S' \supseteq S$, if $A \Vdash_{S'} \bot$ and $B \Vdash_{S'} \bot$ then $\Vdash_{S'} \bot$, so we conclude $\Vdash_{S} A \lor_{c} B$.

     \item[(3, $\Rightarrow$)] Assume $\Vdash_{S} A \to_{c} B$. Then, for all $S' \supseteq S$, if $ \Vdash_{S'} A$ and $B \Vdash_{S'} \bot$ then $\Vdash_{S'} \bot$. Assume, for the sake of proving a contradiction, that there is a $S' \supseteq S$ such that, for all $S'' \supseteq S'$, it holds that $ A\nVdash_{S''} B$. Then for all $S'' \supseteq S'$ there must be a $S''' \supseteq S''$ such that $\Vdash_{S'''} A$ and $\nVdash_{S'''} B$. Pick a maximal consistent extension $S^{M}$ of $S'$. Since $S^{M}$ is only extended by itself, $\Vdash_{S^{M}} A$ and $\nVdash_{S^{M}} B$. But then it also follows that $B \Vdash_{S^{M}} \bot$ is satisfied vacuously, so our first assumption yields $\Vdash_{S^{M}} \bot$. Contradiction. Therefore, for all $S' \supseteq S$ there must be some $S'' \supseteq S'$ such that $A \Vdash_{S''} B$.

      \item[(3, $\Leftarrow$)] Assume that, for all $S' \supseteq S$, there is some $S'' \supseteq S'$ such that $A \Vdash_{S''} B$. Assume, for the sake of proving a contradiction, that there is a $S' \supseteq S$ such that $\Vdash_{S'} A$ and $ B \Vdash_{S'} \bot$. By our first assumption, there must be a $S'' \supseteq S'$ such that $A \Vdash_{S''} B$. Since $\Vdash_{S'} A$ by monotonicity we have $\Vdash_{S''} A$, so by the entailment clause also $\Vdash_{S''} B$. But since $B \Vdash_{S'} \bot$ we conclude $\Vdash_{S''} \bot$ and $\vdash_{S'} \bot$, violating the consistency constraint. So for no $S' \supseteq S$ we have  $\Vdash_{S'} A$ and $ B \Vdash_{S'} \bot$, hence it vacuously holds that, for all $S' \supseteq S$, if  $\Vdash_{S'} A$ and $ B \Vdash_{S'} \bot$ then $\Vdash_{S'} \bot$, hence we conclude $\Vdash_{S} A \to_{c} B$. 
\end{enumerate}
\end{proof}
It is interesting to not that this has direct connections to algebraic structures and to some other collapses shown in the literature on inquisitive logics. When definining his inquisitive intuitionistic logic \cite{Holliday2020-HOLIIL-3}, Holliday points out \footnote{We thank the anonymous reviewer for pointing this source to us.} that a disjunction with inquisitive properties can be obtained by using the clause of Kripke semantics with that of Beth semantics. He also points out that interpreting a declarative disjunction as the regularization of the join in the Heyting algebra of upsets is equivalent to using the following semantic clause:

\begin{quote}
A state $x$ in a poset forces $A \lor B $iff for every $x' \geq x$ there is an $x'' \geq x'$ such that $x''$ forces $A$ or $x''$ forces $B$.
\end{quote}

Since this definition is classical, he argues that intuitionistic declarative disjunctions should use the clause of Beth semantics instead. It is already shown in \cite{Stafford2026-LogicOfPTV} that the Kripke semantics clause collapses into the disjunction of generalised inquisitive logic in a $\Bes$ setting, but since the definition above is just a restatement of part 2 of Theorem 4.11. the present paper has shown that there is a direct algebraic counterpart to our collapse results in inquisitive semantics.

\section{Conclusion}\label{sec:conc}
The study of interactions between classical and intuitionistic notions of proof often unveils connections which are not immediately apparent. The double negation translations of G\"{o}del, Gentzen, Kolmogorov, Kuroda and Krivine \cite{ferreira2025varioustranslationsclassicalintuitionistic} initially showed that classical logic can be faithfully embedded into intuitionistic logic. Prawitz, Krauss and Dowek built on this to show that such embeddings enable a constructive reinterpretation of classical operators, with consequences for philosophy \cite{DBLP:journals/Prawitz15}, computer science \cite{DBLP:journals/Dowek16a} and mathematics \cite{Krauss-P1}. Barroso-Nascimento, Pereira and Pimentel showed in \cite{Pimentel2025-PIMAEV} that the setting of $\Bes$ can be used to provide a framework in which a intuitionistic notion of proof based on Heyting, Brouwer and Dummett's interpretation of proof as constructions \cite{Heyting1956-HEYIAI-2,Brouwer1981-BROBCL,Dummett1977-DUMEOI-2} coexists with Hilbert's interpretation of classical proofs as proofs of consistency \cite{ConsistencyHillbert61554c58-c869-34f0-b322-2cff263d9ae0,Hilbert1979-HILMPL}, which in turn allows the definition of a truly ecumenical semantics for Barroso-Nascimento's ecumenical logic \cite{Barroso-Nascimento}. We extend this result by showing a simple adaptation that turns this semantics into one for Prawitz's ecumenical logic \cite{DBLP:journals/Prawitz15}. Moreover, when one adopts this perspective, a surprising connection between $\Bes$ and Beth's semantics is revealed: despite having originally been stated as a definition of intuitionistic semantics, and despite not being compatible with versions of $\Bes$ that define $\bot$ through a semantic clause in terms of atomic explosion, Beth's semantic clauses and analogous semantic definitions corresponds precisely to the classical operators of Prawitz's ecumenical logic when paired with the version of $\Bes$ that requires all bases to be consistent, thus becoming a definition for classical operators.

The connection between $\Bes$, Beth semantics and ecumenical logic is also interesting due to the similarity of the intended interpretations of $\Bes$ and Beth semantics. The original definitions of Beth models were stated for fans. Those particular frames hold a special relatioship to Brouwer's absolutely free choice sequences, which formalised Brouwer's idea of mathematics as a procedure of accumulation of mental mathematical constructions over time. Similarly, standard interpretations of $\Bes$ view bases as objects that represent states of knowledge \cite{Piecha2015FailComplete-PIEFOC} by specifying which inferences are accepted in a given context, while extensions are states that may be achieved when new inferences are accepted as valid, meaning that the procedure of extending a base represents the accumulation of inferential knowledge. Despite having similar motivations, they cannot be interpreted as equivalent. While fans generally have countably many elements and a strong finitistic interpretation, the set of all atomic bases has uncountably many elements, meaning that $\Bes$ definitions quantify over an uncountable set. This makes it so that the proof of metatheoretical properties sometimes require use of either the Axiom of Choice or equivalent constructions (in our case, Hausdorff's maximal principle). However, just like in \cite{Pimentel2025-PIMAEV}, classical metareasoning is generally required only for proving properties of classical operators, meaning that the metalanguage can once again be interpreted as ecumenical. Interestingly, in this case we have not only a logical ecumenical metalanguage, but also a set-theoretic ecumenical metalanguage, similar to the one originally envisioned by Krauss \cite{Krauss-P1,Krauss-P2}. 

The relationship between $\Bes$, Beth semantics and logical ecumenism warrants further investigation. In future work we intend to investigate how logical ecumenism impacts traditional mathematical definitions in general, but also how they impact constructions related to Beth's semantics in particular. Natural points of investigation are the relationship between logical ecumenism, choice sequences, the Axiom of Choice, $\Bes$, and variants of $\Pts$ compatible with the aforementioned metatheoretical interpretations. We also want to investigate algebraic properties of ecumenical systems and the relationship between our results and inquisitive logics. Such studies are expected to once again unveil unexpected relations between classical and intuitionistic logic

\section*{Acknowledgments}
Barroso-Nascimento, Waddington and Pimentel are supported by the Leverhulme Trust grant RPG-2024-196 and have received funding from the European Union's Horizon 2020 research and innovation programme under the Marie Sk\l{}odowska-Curie grant agreement Number 101299559. The authors would like to thank Luiz Carlos Pereira for introducing ecumenism into their lives. Finally, we would also like to thank the anonymous reviewer for their careful reading and thoughtful comments, which helped us clarify several conceptual points and improve the presentation of the paper.

\bibliographystyle{plain}
\bibliography{ref}

\end{document}